\documentclass{article} 
\usepackage{GEM_workshop_2025,times}

\usepackage{amsmath,amsfonts,bm}

\def\eqref#1{equation~\ref{#1}}

\def\1{\bm{1}}

\DeclareMathAlphabet{\mathsfit}{\encodingdefault}{\sfdefault}{m}{sl}
\SetMathAlphabet{\mathsfit}{bold}{\encodingdefault}{\sfdefault}{bx}{n}

\usepackage{hyperref}
\usepackage{url}
\usepackage{algorithm}
\usepackage{algorithmic}
\usepackage{microtype}
\usepackage{graphicx}
\usepackage{subfigure}
\usepackage{booktabs} 
\usepackage{hyperref}
\usepackage{amsmath}
\usepackage{amssymb}
\usepackage{mathtools}
\usepackage{amsthm}
\usepackage[capitalize,noabbrev]{cleveref}
\theoremstyle{plain}
\newtheorem{theorem}{Theorem}[section]

\newtheorem{lemma}[theorem]{Lemma}

\theoremstyle{definition}
\newtheorem{definition}[theorem]{Definition}

\theoremstyle{remark}
\newtheorem{remark}[theorem]{Remark}
\newcommand{\vect}[1]{\boldsymbol{#1}}

\usepackage[textsize=small]{todonotes}
\usepackage{multirow} 
\usepackage{enumitem}
\usepackage{wrapfig} 
\usepackage{minitoc}
\newcommand*\samethanks[1][\value{footnote}]{\footnotemark[#1]}

\title{Towards More Accurate Full-Atom Antibody Co-Design}

\author{
    Jiayang Wu$^{1}$\thanks{Equal contribution.} \quad Xingyi Zhang$^{2}$\samethanks[1]
    \quad Xiangyu Dong$^{3}$ \quad Kun Xie$^{3}$ \quad \textbf{Ziqi Liu}$^{4}$ \\ \hspace{0.1mm}  \textbf{Wensheng Gan}$^{1}$ \quad \textbf{Sibo Wang}$^{3}$ \quad \textbf{Le Song}$^{2}$ \\
    $^{1}$Jinan University \quad $^{2}$MBZUAI \quad $^{3}$The Chinese University of Hong Kong \\ 
    $^{4}$ University of Chinese Academy of Sciences\\
\hspace{0.1mm} \texttt{xingyi.zhang@mbzuai.ac.ae} \\
}

\iclrfinalcopy
\begin{document}
\doparttoc 
\faketableofcontents 
\part{}

\maketitle

\begin{abstract}
Antibody co-design represents a critical frontier in drug development, where accurate prediction of both 1D sequence and 3D structure of complementarity-determining regions (CDRs) is essential for targeting specific epitopes. Despite recent advances in equivariant graph neural networks for antibody design, current approaches often fall short in capturing the intricate interactions that govern antibody-antigen recognition and binding specificity. In this work, we present Igformer, a novel end-to-end framework that addresses these limitations through innovative modeling of antibody-antigen binding interfaces. Our approach refines the inter-graph representation by integrating personalized propagation with global attention mechanisms, enabling comprehensive capture of the intricate interplay between local chemical interactions and global conformational dependencies that characterize effective antibody-antigen binding. Through extensive validation on epitope-binding CDR design and structure prediction tasks, Igformer demonstrates significant improvements over existing methods, suggesting that explicit modeling of multi-scale residue interactions can substantially advance computational antibody design for therapeutic applications.
\end{abstract}

\section{Main}

Antibodies are Y-shaped proteins that play a pivotal role in the immune system, capable of binding with high specificity to antigens, including pathogens and foreign molecules \cite{murphy2016janeway}. The design of antibodies for specific epitopes in antigens is an essential yet challenging task with broad applications in therapeutic development and drug discovery \cite{feldmann2003tnf,adams2005cancer,mullard2022fda,wang2024human}. 
The challenges arise primarily from the variability and structural complexity of complementarity-determining regions (CDRs), where binding primarily occurs \cite{he2024novo}. Understanding and modeling the intricate interactions between antibodies and antigens is crucial for developing an effective computational antibody design approach \cite{baran2017principles,tennenhouse2024computational,yan2024antibodies}.

\begin{wrapfigure}{r}{0.42\textwidth}
  \centering
  \hspace{-4mm}
  \includegraphics[width=\linewidth]{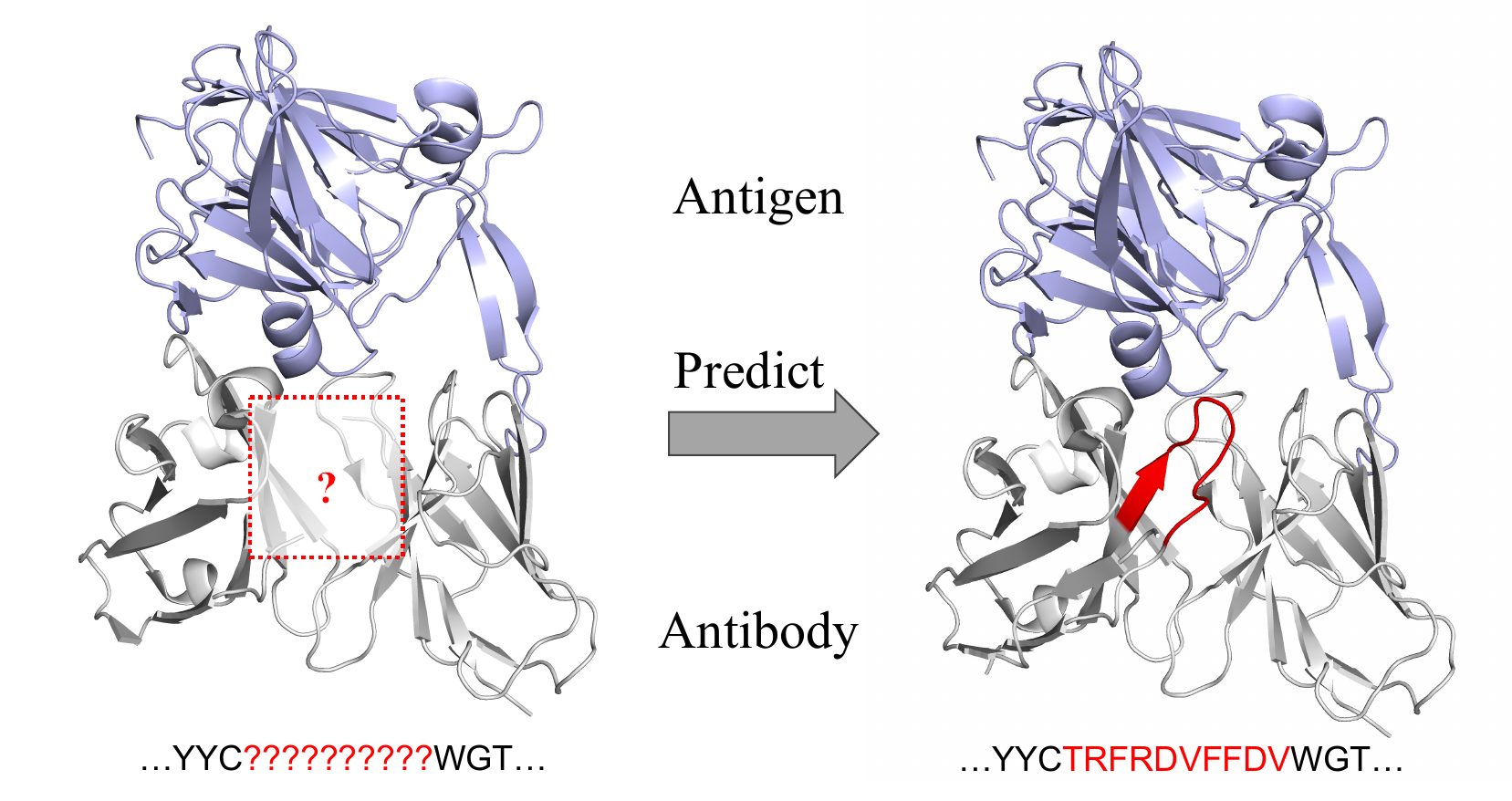}  \vspace{-3mm}
  \caption{End-to-end antibody co-design task.}  \label{fig:fig2-task}
  \vspace{-3mm}
\end{wrapfigure}

The past decade has witnessed a significant evolution in antibody design approaches. Traditional methods, including sequence-based models \cite{alley2019unirep,shin2021protein} and energy-based optimization frameworks \cite{pantazes2010optcdr,lapidoth2015abdesign,adolf2018rabd}, provided initial solutions but fell short of fully capturing the structural and functional interactions between antibodies and antigens \cite{ingraham2019generative}. 
More recently, learning-based approaches, particularly equivariant graph neural networks, deep generative models, and diffusion models \cite{saka2021antibody,jin2022refinegnn,watson2022broadly,luo2022diffab,kong2023dymean,martinkus2023abdiffuser,wang2024iggm}, have emerged as powerful tools for simultaneously co-designing CDR 1D sequences and 3D structures.
These models represent a significant advancement over traditional approaches by leveraging deep learning to model complex sequential and structural relationships. However, many of these methods simplify the problem by focusing exclusively on backbone modeling or specific CDR regions, overlooking side-chain interactions and the full-atom geometry that are crucial to accurately modeling antibody-antigen binding \cite{ruffolo2023fast,martinkus2023abdiffuser}.

Traditional computational antibody design relies on multi-stage pipelines that separately handle structure prediction, docking, and CDR generation. A recent work, dyMEAN \cite{kong2023dymean}, integrates these tasks into a unified end-to-end framework, providing a more streamlined approach to antibody co-design. However, despite these advances, dyMEAN does not fully take advantage of intricate antibody-antigen interactions, leaving opportunities for further improvement in capturing the detailed molecular interplay at binding interfaces. These limitations highlight the need for a more sophisticated approach that can simultaneously handle full-atom geometry and complex hierarchical interactions in antibody-antigen complexes. Further discussion of related works can be found in Appendix~\ref{appendix-sec:related-work}.

To address these limitations, in this work, we introduce Igformer\footnote{\em \underline{I}mmuno\underline{g}lobulin Trans\underline{former} (Igformer)}, a novel end-to-end framework for antibody co-design that advances antibody co-design through three key innovations. At its core, Igformer introduces a sophisticated antibody-antigen inter-graph refinement strategy that uniquely integrates personalized propagation with a global attention mechanism. 
The personalized propagation scheme models local chemical properties and geometric constraints, while the complementary global attention mechanism captures long-range structural dependencies. This novel dual-scale architecture enables Igformer to maintain atomic-level precision while effectively capturing intricate residue-level interactions within the antibody-antigen binding interface. Furthermore, Igformer incorporates E(3)-equivariance throughout the framework, ensuring that all structural predictions adhere to the physical and geometric constraints of biological systems.

Through comprehensive experimental evaluation, Igformer demonstrates superior effectiveness across multiple challenging antibody design tasks. For epitope-binding CDR design, Igformer shows a 2.02\% improvements over state-of-the-art models in terms of amino acid recovery rate. For antigen-antibody complex structure prediction, Igformer achieves remarkable accuracy with an 11.84\% reduction in RMSD compared to existing approaches. These substantial improvements validate the effectiveness of our technical innovations in addressing the challenges of full-atom geometry and sequence-structure co-design. Igformer thus represents a significant advancement in computational antibody design, offering promising new directions for accelerating drug discovery and therapeutic development.

\begin{figure*}[t]
\vspace{-2mm}
  \centering
  \includegraphics[width=\textwidth]{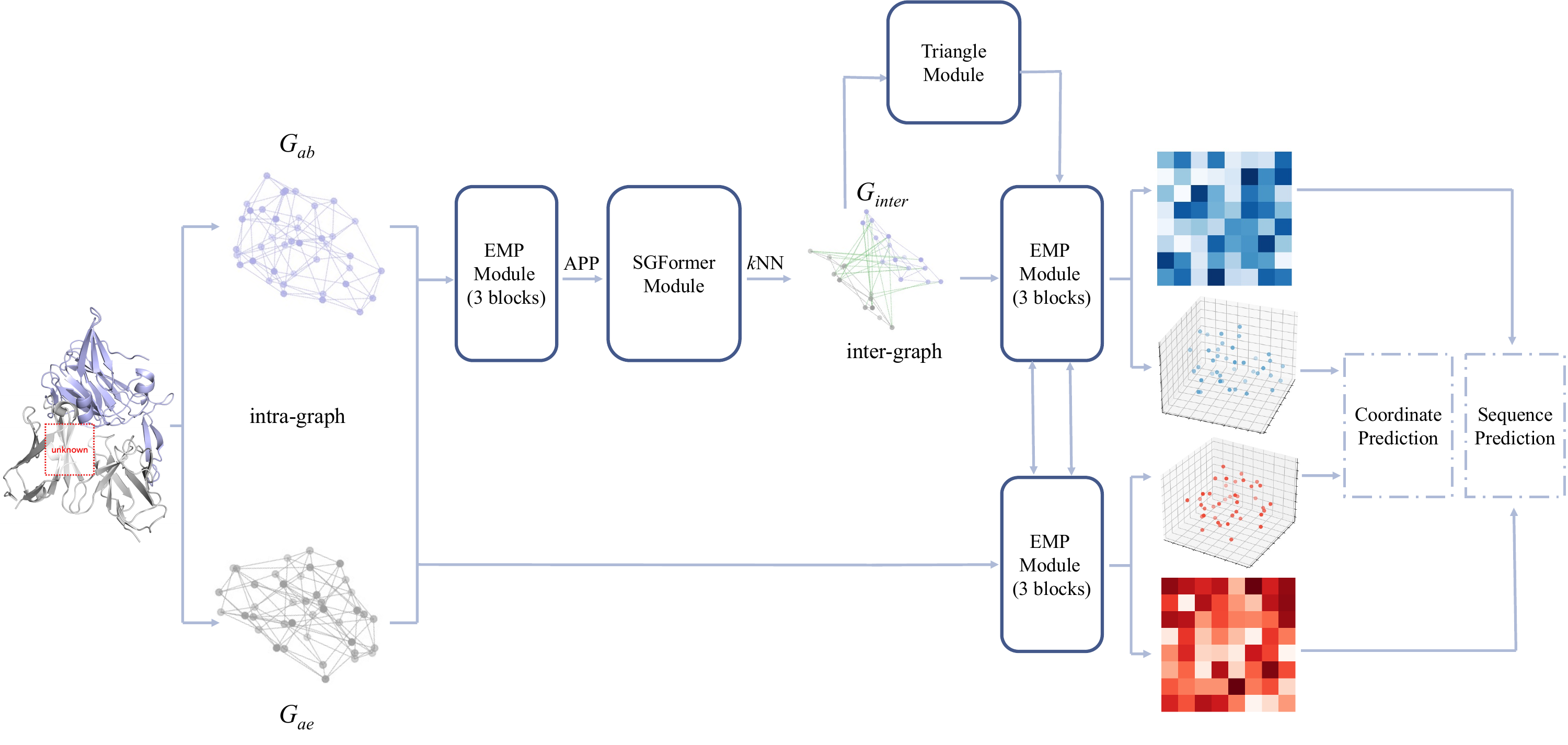}
  \vspace{-8mm}
  \caption{Framework of Igformer.}
  \label{fig:fig3-framework}
  \vspace{-4mm}
\end{figure*}

\section{Results}
\label{sec:experiment}
\subsection{IgFormer}
Figure~\ref{fig:fig3-framework} illustrates the overall framework of Igformer. Specifically, the {\em Equivariant Message Passing (EMP)} modules serve as foundational blocks of Igformer, integrating both spatial and biochemical properties during the message-passing process. Our theoretical analysis confirms the E(3)-equivariant property of EMP, which lays the theoretical foundation of Igformer.

The key innovation of Igformer lies in its sophisticated inter-graph refinement strategy. The process begins with the initialization of residue representations and coordinates, followed by intra-graph construction. Subsequently, based on the constructed intra-graph, Igformer captures the complex interactions within the binding range between antibody paratope and antigen epitope through two key modules: {\em Approximate Personalized Propagation (APP)}, which preserves local interaction information, and {\em Simplified Graph Transformer (SGFormer)}, which captures global binding dependencies. Next, the learned intricate binding patterns are utilized to refine the inter-graph, followed by {\em Triangle Multiplicative (TM)} module and {\em Axial Attention (AA)} module, which generate pairwise residue interactions through different mechanisms. Finally, the refined inter-graph, together with the intra-graph are processed through a {\em dual-scale EMP modules}, responsible for the final coordinate and sequence predictions. Systematical elaboration on our method and supplementary details can be found in Appendix~\ref{sec:framework} and appendix~\ref{appendix-sec:sup}.

{\bf Dataset.}
Following previous works on antibody design~\cite{jin2022hern,kong2023dymean,kong2023mean}, we train all models on the \textit{Structural Antibody Database} (SAbDab)~\citep{dunbar2013sabdab,schneider2021sabdab} using the snapshot from November 2022. We split the dataset into training and validation sets with a ratio of 9:1, yielding $3,246$ antibodies for training and $365$ antibodies for validation.
The antibodies are clustered based on their CDR-H3 regions using a $40\%$ sequence identity threshold, calculated using the BLOSUM62 substitution matrix \cite{foote1992antibody}. This clustering process, performed using MMseqs2, results in $1,644$ clusters in the training set and $182$ clusters in the validation set.
After that, we evaluate Igformer and other competitors on the RAbD benchmark \cite{adolf2018rabd} for Tasks 1-3, which consists of 60 diverse antibody-antigen complexes. For Task 4, Igfold benchmark \cite{ruffolo2023IgFold} consisting of $51$ antibody-antigen complexes is utilized. This test set selection prevents data leakage during the evaluation phase.
Detailed introduction to datasets and training settings of IgFormer are listed in Appendix~\ref{appendix-sec:dataset} and Appendix~\ref{appendix-sec:exp-setting}, respectively.

\begin{table}[t]
\vspace{-2mm}
\caption{CDR-H3 design. Results of models with * are collected from the dyMEAN paper.}
\small
\centering
\label{tab:task1-main}
\begin{tabular}{lccc|ccc}
\hline \hline
\multirow{2}{*}{Model} & \multicolumn{3}{c|}{Generation} & \multicolumn{3}{c}{Docking} \\
\cline{2-7}
& AAR$\uparrow$ & TMscore$\uparrow$ & IDDT$\uparrow$ & CAAR$\uparrow$ & RMSD$\downarrow$ & DockQ$\uparrow$ \\
\hline
RosettaAb*  & 32.31\% & 0.9717 & 0.8272 & 14.58\% & 17.70 & 0.137 \\
DiffAb*     & 35.31\% & 0.9695 & 0.8281 & 22.17\% & 23.24 & 0.158 \\
MEAN*       & 37.38\% & 0.9688 & 0.8252 & 24.11\% & 17.30 & 0.162 \\
HERN*        & 32.65\% & -      & -      & 19.27\% & 9.15 & 0.294 \\
dyMEAN      & 42.64\% & 0.9728 & 0.8438 & 27.35\% & 8.42 & 0.408 \\
\hline
Igformer    & \textbf{43.50\%} & \textbf{0.9757} & \textbf{0.8650} & \textbf{28.11\%} & \textbf{7.15} & \textbf{0.450} \\
\hline \hline
\end{tabular}
\vspace{-4mm}
\end{table}

\begin{table}[t]
\small
\centering
\caption{AAR (\%) on multiple CDRs design.}
\label{tab:task2-main1}
\begin{tabular}{l|cccccc}
\hline \hline
Model & CDR-L1 & CDR-L2 & CDR-L3 & CDR-H1 & CDR-H2 & CDR-H3 \\
\hline
dyMEAN    & \textbf{75.55} & 83.10 & 52.12 & 75.51 & 68.48 & 37.53 \\
Igformer  & 75.20 & \textbf{85.32} & \textbf{63.42} & \textbf{77.20} & \textbf{69.25} & \textbf{41.10} \\
\hline \hline
\end{tabular}
\end{table}

\begin{figure}[t]
\centering
\vspace{-2mm}
  \begin{small}
    \begin{tabular}{cc}
        \includegraphics[width=0.32\linewidth]{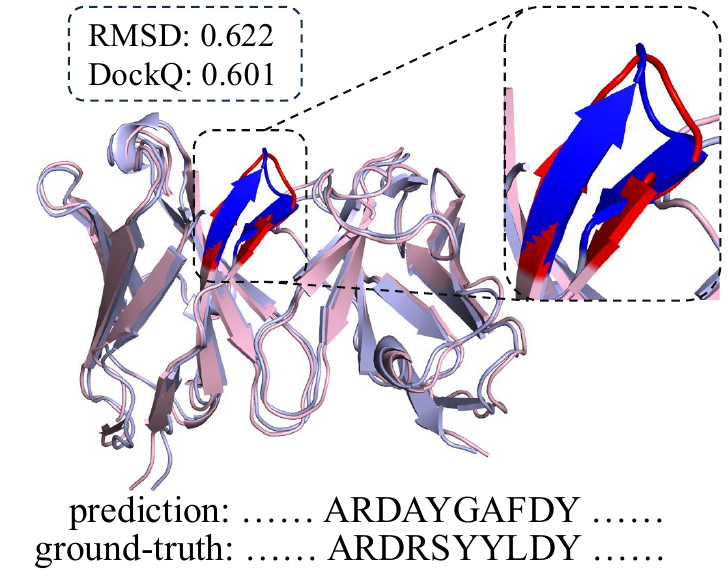} &
        \hspace{-3mm}
        \includegraphics[width=0.32\linewidth]{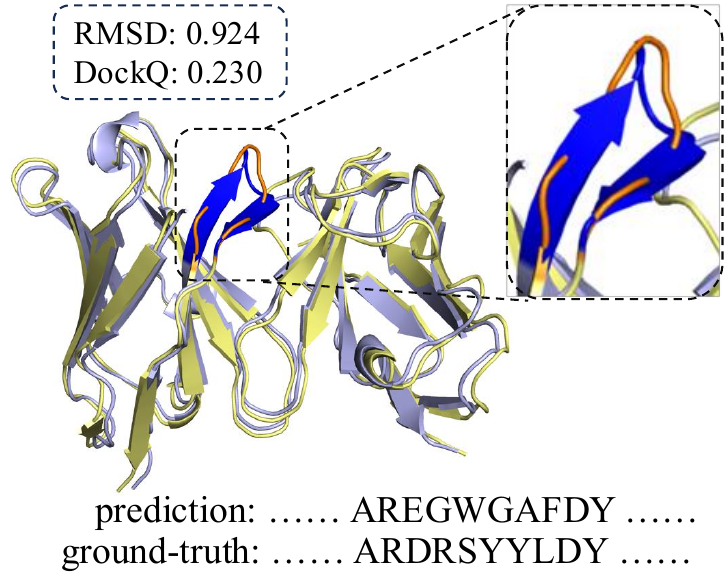} \\
        (a) Igformer &
        \hspace{-2mm}
        (b) dyMEAN \\
    \end{tabular}
    \vspace{-2mm}
    \caption{Antibody structures by Igformer \& dyMEAN.}
    \label{fig:fig4-task1-case}
    \vspace{-4mm}
  \end{small}
\end{figure}

\subsection{Task 1: CDR-H3 Design}
We first evaluate each model for both sequence and structure prediction on CDR-H3, which plays a critical role in antigen binding. During training, the model simultaneously predicts CDR-H3 residue sequences and generates coordinates for the entire antibody structure. 
We compare Igformer against five SOTA methods: RosettaAb \cite{adolf2018rabd}, DiffAb \cite{luo2022diffab}, HERN \cite{jin2022hern}, MEAN \cite{kong2023mean}, and dyMEAN \cite{kong2023dymean}. Evaluation metrics include unaligned RMSD for CDR-H3 using CA atoms and sequence-based metrics for CDR-H3. Details of baselines and metrics can be found in Appendix~\ref{appendix-sec:baselines} and Appendix~\ref{appendix-sec:metrics}, respectively.

Table \ref{tab:task1-main} reports the results of the epitope-binding CDR-H3 design task on the RAbD dataset. As we can observe, Igformer demonstrates superior performance across all evaluation metrics. Specifically, it achieves an AAR of 43.50\%, representing a 2.2\% relative improvement over previous SOTA method dyMEAN, and substantially outperforms earlier approaches like RosettaAb and HERN. Most notably, Igformer excels in docking performance, achieving an RMSD of 7.15 and a DockQ score of 0.450, marking a relative improvement of 15.08\% and 10.29\% respectively over the second-best method. Figure~\ref{fig:fig4-task1-case} provides a comparative visualization of predictions generated by Igformer and dyMEAN. These results demonstrate the enhanced capability of Igformer to capture antibody-antigen interactions and generate more accurate binding interface predictions.

\subsection{Task 2: Multiple CDR Design}
In this set of experiments, we compare Igformer against dyMEAN, the current leading method, on sequence and structure prediction across all six CDRs using the RAbD dataset. The evaluation examines structural accuracy through AAR for individual CDR regions and assesses overall antibody structure quality using additional metrics.
As shown in Tables \ref{tab:task2-main1}-\ref{tab:task2-main2-task3}, Igformer achieves better sequence recovery results in 5 out of 6 CDRs except for CDR-L1, where Igformer maintains comparable performance with dyMEAN. Moreover, Igformer consistently surpasses dyMEAN in overall performance by substantial margins. Notably, Igformer achieves relative improvements of 5.83\% in overall AAR and 21.24\% in DockQ score. These comprehensive results validate the effectiveness of Igformer in both sequence prediction and structural generation across multiple CDR regions.

\subsection{Task 3: Full Antibody Design}
\begin{table}[t]
\small
\vspace{-2mm}
\centering
\caption{Comparison on multiple CDR and full antibody design.}
\label{tab:task2-main2-task3}
\begin{tabular}{l|cccc}
\hline \hline
\multicolumn{5}{c}{Multiple CDR Design} \\
\hline
Model & AAR$\uparrow$ & TMscore$\uparrow$ & lDDT$\uparrow$ & DockQ$\uparrow$ \\
\hline
dyMEAN & 60.05\% & 0.9654 & 0.8029 & 0.3973 \\
Igformer & \textbf{63.55\%} & \textbf{0.9750} & \textbf{0.8311} & \textbf{0.4817} \\
\hline \hline
\multicolumn{5}{c}{Full Antibody Design} \\
\hline
Model & AAR$\uparrow$ & TMscore$\uparrow$ & lDDT$\uparrow$ & DockQ$\uparrow$ \\
\hline
dyMEAN & 71.37\% & 0.9662 & 0.7471 & 0.4237 \\
Igformer & \textbf{73.69\%} & \textbf{0.9681} & \textbf{0.7580} & \textbf{0.4600} \\
\hline \hline
\end{tabular}
\vspace{-4mm}
\end{table}

In this set of experiments, we evaluate predictions across all regions, including both frameworks and CDRs on RAbD dataset, where the sequences and coordinates for the entire antibody are masked in the testing phase. Table \ref{tab:task2-main2-task3} presents the performance comparison for full antibody design. 
These results align with the fundamental characteristic of antibody structure: framework regions exhibit higher conservation across different antibodies, making them more predictable than hyper-variable CDR regions. Our experimental findings confirm this biological characteristic, as both models achieve higher sequence recovery rates in full antibody design compared to CDR-specific prediction tasks.
Notably, Igformer demonstrates consistent superiority over dyMEAN across all evaluation metrics. Igformer achieves a 3.25\% relative improvement in AAR and an 8.57\% relative enhancement in DockQ score. These significant improvements indicate that Igformer not only generates more accurate sequence predictions but also maintains higher structural fidelity, leading to improved antigen binding affinity prediction.

\subsection{Task 4: Complex Structure Prediction}
In this set of experiments, we evaluate the model performance on the Igfold benchmark \cite{ruffolo2023IgFold}, focusing on complex structure prediction when given complete antibody sequences, including CDR-H3. Evaluation metrics include CDR-H3-specific RMSD and comprehensive structural assessments of the entire antibody complex. 

As demonstrated in Table \ref{tab:task4-main}, Igformer establishes new state-of-the-art performance in the complex structure prediction task, particularly excelling in docking metrics. Specifically, Igformer achieves relative improvements of 12.44\% in RMSD and 15.49\% in DockQ scores, respectively. Notably, even when compared to GT$\Rightarrow$HERN, which utilizes ground-truth structures for docking, Igformer demonstrates superior docking performance. These results validate the effectiveness of the proposed inter-graph refinement strategy in capturing the intricate residue interactions within antibody-antigen binding interfaces.

\begin{table}[t]
\small
\centering
\caption{Comparison on complex structure prediction. Results with $^{\dagger}$ are predicted using the ground-truth (GT) structure and serve as the upper bound of this task.}

\label{tab:task4-main}
\begin{tabular}{lcc|cc}
\hline \hline
\multirow{2}{*}{Model} & \multicolumn{2}{c}{Structure} & \multicolumn{2}{c}{Docking} \\
\cline{2-5}
 & TMscore$\uparrow$ & lDDT$\uparrow$ & RMSD$\downarrow$ & DockQ$\uparrow$ \\
\hline
IgFold$\Rightarrow$HDock* & 0.9701 & 0.8439 & 16.32 & 0.202 \\
IgFold$\Rightarrow$HERN* & 0.9702 & 0.8441 & 9.63 & 0.429 \\
GT$\Rightarrow$HERN & 1.0000$^{\dagger}$ & 1.0000$^{\dagger}$ & 9.65 & 0.432 \\
dyMEAN & \textbf{0.973} & 0.8676 & 9.00 & 0.452 \\
Igformer & \textbf{0.973} & \textbf{0.8677} & \textbf{7.88} & \textbf{0.522} \\
\hline \hline
\end{tabular}
\end{table}
\begin{table}[]
\caption{Ablation study on CDR-H3 design.}
\small
\centering
\label{tab:ablation-main}
\begin{tabular}{lccc|ccc}
\hline \hline
\multirow{2}{*}{Model} & \multicolumn{3}{c}{Generation} & \multicolumn{3}{c}{Docking} \\
\cline{2-7}
 & AAR$\uparrow$ & TMscore$\uparrow$ & lDDT$\uparrow$ & CAAR$\uparrow$ & RMSD$\downarrow$ & DockQ$\uparrow$ \\
\hline
Igformer & \textbf{43.50\%} & \textbf{0.9757}  & \textbf{0.8650} & \textbf{28.11\%}  & \textbf{7.15} & \textbf{0.450} \\
- APP & 42.80\% & 0.9725 & 0.8412 & 27.60\% & 8.30 & 0.416 \\
- SGFormer & 43.45\% & 0.9743 & 0.8578 & 28.00\% & 7.82 & 0.431 \\
- TM & 43.32\% & 0.9731 & 0.8499 & 27.91\% & 8.40 & 0.417 \\
- AA & 43.06\% & 0.9735 & 0.8460 & 27.66\% & 8.12 & 0.425 \\
- Dual EMP & 42.75\% & 0.9728 & 0.8430 & 27.50\% & 8.22 & 0.419 \\
\hline \hline
\end{tabular}
\end{table}

\subsection{Ablation Study}
We conduct comprehensive ablation studies to evaluate the individual contributions of key components in Igformer. Specifically, we systematically analyze APP and SGFormer introduced in Section~\ref{sec:subsec-inter-graph}, triangle multiplicative module (TM), axial attention (AA), and dual EMP module presented in Section~\ref{sec:subsec-update}. Table~\ref{tab:ablation-main} illustrates the impact of removing each component from Igformer on epitope-binding CDR-H3 prediction tasks. Results of experiments on Task 4 are provided in Appendix~\ref{appendix-sec:subsec-ablation}.

Our observations are as follows. Removing the APP reduces performance across all metrics, particularly affecting the DockQ score (0.416) and RMSD (8.30). The removal of SGFormer shows similar trends but with less impact, achieving a DockQ score of 0.431 and RMSD of 7.82. These results demonstrate their crucial role in capturing complicated interface interactions. 
The triangle multiplicative module and axial attention mechanism both contribute significantly to model performance. Without TM, the AAR drops to 43.32\% and DockQ to 0.417, while removing AA results in an AAR of 43.06\% and DockQ of 0.425. Most notably, replacing our dual EMP architecture with a single message passing framework results in significant performance degradation across all metrics, with AAR dropping to 42.79\% and DockQ to 0.417, validating our design choice of separate intra- and inter-graph message passing processes.

\section{Conclusion}
\label{sec:conclusion}
In this paper, we propose Igformer, an end-to-end antibody co-design framework that simultaneously generates both 1D sequences and 3D structures for epitope-binding antibody CDRs. Igformer advances existing models through a novel intra-graph refinement process, which is capable of capturing intricate residue interactions within the epitope-paratope binding region. Through extensive experimental evaluation, we demonstrate the effectiveness of Igformer, providing a new perspective for end-to-end antibody co-design tasks.

\bibliography{igformer}
\bibliographystyle{iclr2025_conference}

\newpage
\appendix
\onecolumn
\addcontentsline{toc}{section}{Appendix} 
\part{} 
\parttoc 

\newpage
\section{Related Work}
\label{appendix-sec:related-work}
{\bf Antibody Co-design.}
Antibody co-design methodologies facilitate the simultaneous generation of sequence and structure in an end-to-end manner, thereby enabling the modeling of intricate dependencies between backbone conformation and amino acid composition. Specifically, RefineGNN \cite{jin2022refinegnn} utilized an iterative refinement approach, alternately predicting atom coordinates and residue types in CDRs through auto-regression. DiffAb \cite{luo2022diffab} advanced the field by introducing a diffusion-based framework that explicitly targets specific antigens while considering atomic-level structures, enabling both sequence and structure generation for CDRs based on framework regions and target antigens. MEAN \cite{kong2023mean} and its successor dyMEAN \cite{kong2023dymean} leveraged graph neural networks with E(3)-equivariant architectures for CDR prediction, with dyMEAN extending to full-atom geometry modeling including both backbone and side chains. AbDiffuser \cite{martinkus2023abdiffuser} introduced an efficient SE(3)-equivariant architecture called APMixer for complete antibody structure generation, demonstrating success in wet-lab validation. Most recently, IgGM \cite{wang2024iggm} expanded capabilities to end-to-end antibody design through a hybrid diffusion model that can generate complete antibody sequences and structures simultaneously without requiring template structures.Despite these advances, existing methods have largely simplified or overlooked the complex interactions at antibody-antigen binding interfaces. Our work addresses this limitation through a novel inter-graph refinement strategy that integrates personalized propagation into global attention mechanisms, enabling more accurate modeling of binding interface dynamics.

{\bf Sequence Design.}
The field of sequence design has evolved from LSTM-based approaches to more advanced language models. LSTM-based antibody design \cite{saka2021antibody} pioneered sequence generation for affinity maturation by employing a long short-term memory network trained on enriched phage display sequences to generate and prioritize antibody variants, 
using likelihood scores from the trained model to identify promising candidates and outperforming traditional frequency-based screening approaches. ProGen \cite{madani2020progen} introduced a controllable protein generation language model trained on 280M sequences, enabling both high-quality sequence generation and fine-grained control through conditioning. The model employed a transformer architecture and demonstrated the ability to generate proteins with native-like structural and functional properties. 
More recently, EvoDiff \cite{alamdari2023protein} introduced an evolutionary-scale diffusion framework operating directly in sequence space. The model combines discrete diffusion with evolutionary-scale sequence data to enable both unconditional generation of diverse, structurally-plausible proteins and flexible conditional generation via sequence inpainting.

{\bf Structure Design. } 
Structure design has progressed from single-step prediction to sophisticated diffusion-based approaches. RFjoint \cite{wang2021RFjoint} initially pioneered functional site scaffolding with structure prediction networks, though limited by its single-step prediction approach. RFdiffusion \cite{watson2022RFdiffusion} advanced this approach by fine-tuning the RoseTTAFold structure prediction network as a denoising network within a diffusion model framework, enabling iteratively protein structure building from random noise.
IgFold \cite{ruffolo2023IgFold} introduced a specialized antibody structure prediction framework combining a pre-trained language model (AntiBERTy) with graph neural networks. This innovation enabled direct prediction of backbone atom coordinates, achieving accuracy comparable to AlphaFold 2 \cite{jumper2021alphafold} while being significantly faster and eliminating the need for multiple sequence alignments or template structures.
Most recently, AlphaFold 3 \cite{abramson2024AlphaFold3} has taken a broader approach by introducing a unified deep-learning framework capable of predicting structures for diverse biomolecular complexes. It achieves this by replacing evoformer in Alphafold2 with a simpler Pairformer and employing a diffusion-based structure module for direct atom coordinates prediction, representing a significant step toward a more generalized solution for biomolecular structure prediction.

\section{Problem Formulation}
\label{sec:preliminary}
Antibodies are Y-shaped proteins with a distinctive molecular architecture: they comprise heavy and light chains, each containing constant domains that are conserved across antibodies and variable domains that determine binding specificity. Within these variable domains, the structure alternates between framework regions (FRs) that maintain structural stability and {\em complementarity-determining regions (CDRs)} that are responsible for antigen recognition, with CDR positions defined by the IMGT numbering scheme \cite{lefranc2003imgt}. Among these regions, the CDR-H3 loop in the heavy chain variable domain exhibits the highest variability and plays a pivotal role in antigen recognition \cite{narciso2011analysis, tsuchiya2016diversity, frederic2023abmpnn}. In the context of antibody-antigen interactions, the binding surface of the antibody is termed the {\em paratope}, while its corresponding target region on the antigen is called the {\em epitope}, as illustrated in Figure~\ref{fig:fig1-complex}.

\begin{figure}[t]
  \centering
  \hspace{-2mm}
  \includegraphics[width=0.8\textwidth]{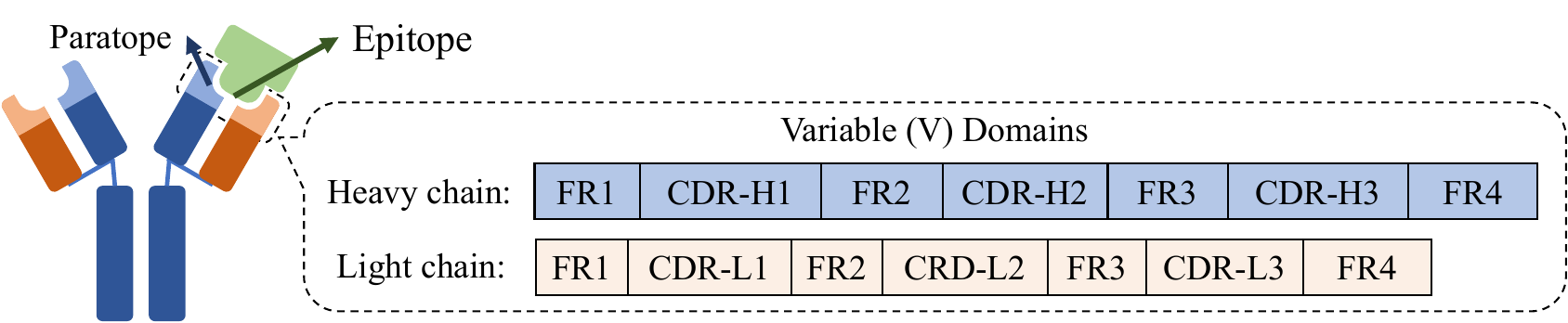}
  \vspace{-3mm}
  \caption{Illustration of antibody-antigen complex.}
  \vspace{-3mm}
  \label{fig:fig1-complex}
  \vspace{-2mm}
\end{figure}

The objective of antibody co-design is to simultaneously predict the amino acid sequence and 3D structure of antibody CDRs that optimally bind a given epitope on a target antigen, as illustrated in Figure~\ref{fig:fig2-task}. 
We formulate this problem using a dual-graph representation of the binding region. The first component is an antibody graph $G_{ab} = \{V_{ab}, E_{ab}\}$, while the second is an antigen epitope graph $G_{ae} = \{V_{ae}, E_{ae}\}$. The vertices $V_{ab}$ and $V_{ae}$ represent amino acid residues, where each residue $v_i$ is characterized by its amino acid type $s_i$ and a coordinate matrix $\vect{X}_i \in \mathbb{R}^{3\times c_i}$ encoding the 3D positions of its $c_i$ atoms, including both backbone and side chain atoms. Spatial proximities between residues are captured through edges $E_{ab}$ and $E_{ae}$ constructed using a {\em $k$-nearest neighbors (kNN)} approach over all atomic positions. The complete graph construction process is detailed in Appendix~\ref{appendix-sec:graph-construct}. Within this framework, we denote the paratope residues as a subset $V_p \subseteq V_{ab}$ of the antibody graph, where we specifically focus on {\em CDR-H3} as the paratope region. 
\begin{definition}
\label{def:antibody-codesign}
    Given partial antibody sequence $\{ s_i \mid v_i\in V_{ab}, v_i\notin V_p \}$ and the target antigen epitope sequence $\{s_j \mid v_j\in V_{ae}\}$ with corresponding structural information $\vect{X}_{ab}$ and $\vect{X}_{ae}$, the goal of antibody co-design is to predict both the 1D sequence $\{s_k \mid v_k \in V_p\}$ of the paratope and 3D structure $\{\vect{X}_k \mid v_k\in V_p\}$ of the entire antibody in the context of the antibody-antigen complex.
\end{definition}
This formulation defines a joint optimization problem where the model simultaneously predicts both the amino acid sequence of the antibody paratope and its complete 3D coordinates within the antibody-antigen binding region. Through this formulation, we aim to harness the expressive power of equivariant graph neural architectures to capture the intricate spatial and chemical relationships that govern antibody-antigen interactions, enabling end-to-end structure-based antibody co-design.

\section{Method}
\label{sec:framework}

In this section, we present Igformer, an end-to-end framework for antibody co-design that introduces a novel approach to modeling antibody-antigen interactions. At its core, Igformer advances previous attempts through a sophisticated inter-graph refinement mechanism that combines personalized propagation with global attention, enabling more accurate modeling of complex interactions within antibody-antigen binding interfaces.
As illustrated in Figure~\ref{fig:fig3-framework}, Igformer consists of three key components: (i) Equivariant Message Passing module, (ii) Inter-graph Refinement module, and (iii) Iterative Update module. In the following, we will detail the initialization process in Section~\ref{sec:subsec-init}, elaborate on the Equivariant Message Passing module in Section~\ref{sec:subsec-emp}, followed by the Inter-Graph Refinement module in Section~\ref{sec:subsec-inter-graph}, and the Iterative Update module in Section~\ref{sec:subsec-update}. We then present the Igformer pipeline in Section~\ref{sec:subsec-pipeline} and learning objectives in Section~\ref{sec:subsec-loss}.

\subsection{Initialization}
\label{sec:subsec-init}
The foundation of Igformer lies in its comprehensive representation scheme that integrates both biochemical and structural information through embedding and coordinate initialization steps.

{\bf Feature initialization.}
For each residue $v_i$, we construct a dual-component embedding: a residue embedding $\vect{H}_i^{res} \in \mathbb{R}^{d}$ that encodes amino acid properties, and a position embedding $\vect{H}_i^{pos} \in \mathbb{R}^{d}$ that captures sequential context through distinct indices for antigen, heavy chain, and light chain regions. These components are combined additively to form the final residue representation:
$
    \vect{H}_i = \vect{H}_i^{res} + \vect{H}_i^{pos},
$
This initialization approach enables the simultaneous processing of amino acid properties and their structural context while differentiating between functionally distinct antibody-antigen complex regions. More details of embedding initialization are provided in Appendix~\ref{appendix-sec:embeds}.

{\bf Coordinate initialization.}
Igformer employs a structured coordinate initialization strategy that combines template-based modeling with precise alignment and normalization procedures. For the epitope region $V_{ae}$, we maintain the actual experimental coordinates, while other regions are initialized using templates (refer to Appendix~\ref{appendix-sec:subsec-template-coord}) that provide positions of backbone atoms (N, CA, C, O). Missing template coordinates are estimated through linear interpolation between the nearest known template positions along the sequence.
We then optimize these initial structures through rigid-body alignment using the Kabsch algorithm \cite{kabsch1976rotation}, which minimizes structural discrepancies between the template and known coordinates.
Following initialization, we apply a chain-specific processing pipeline to establish a consistent reference frame. Coordinates of each chain are centered relative to their respective virtual center of mass, and subsequently normalized to ensure uniform scale across the complex.
This carefully calibrated initialization process provides a robust foundation for subsequent equivariant message passing and inter-graph refinement stages. For a comprehensive description of the coordinate initialization and processing procedures, we refer readers to Appendix~\ref{appendix-sec:coord}.

\subsection{Equivariant Message Passing}
\label{sec:subsec-emp}

Given input coordinates $\vect{X}$ and embeddings $\vect{H}$ of residues in a graph $G$, at the $l$-th layer, the {\em equivariant message passing (EMP)} module updates residue representations and coordinates 
$\vect{H}_i^{(l+1)}, \vect{X}_i^{(l+1)} = \textit{EMP}\left(\vect{H}_i^{(l)}, \vect{X}_i^{(l)}\right)$ through four sequential steps.

First, we compute multi-level pairwise residue similarities by combining residue and atomic information:
\begin{equation*}
\begin{aligned}
    \Delta \vect{X}_{ij}^{(l)} &= \vect{X}_{i}^{(l)} - \vect{X}_{j}^{(l)}, \ sim_{ij}^{res} = \Delta \vect{X}_{ij}^{(l)}(\Delta \vect{X}_{ij}^{(l)})^{\top}, \\
    sim_{ij(m,n)}^{atom} &= \sum_{c=1}^3 (\vect{X}_{i,m,c} - \vect{X}_{j,n,c})(\vect{X}_{i,m,c} - \vect{X}_{j,n,c})^{\top},
\end{aligned}
\end{equation*}
where $m, n \in \{1,\cdots,14\}$\footnote{Each residue is represented using 14 atoms in 3D space.} denote different atoms in the residue. The multi-level positional information is processed through MLPs and combined into a similarity matrix:
\begin{equation*}
    \vect{S}_{i,j}^{(l)} = w \cdot \textit{MLP}_1(sim_{ij}^{res}) + (1-w) \cdot \textit{MLP}_2(sim_{ij}^{atom}).
\end{equation*}
Next, we update edge features by incorporating node information and similarity scores:
\begin{equation*}
    \vect{H}_{e_{ij}}^{(l+1)} = \textit{EdgeMLP} \left(  \vect{H}_i^{(l)} \oplus \vect{H}_j^{(l)} \oplus \vect{S}_{i,j}^{(l)} \oplus \vect{H}_{e_{ij}}^{(l)}  \right),
\end{equation*}
where $\oplus$ denotes concatenation operation. The coordinates are subsequently updated through weighted aggregation of neighbor differences:
\begin{equation*}
    \vect{X}_i^{(l+1)} = \vect{X}_i^{(l)} + \sum_{j \in \mathcal{N}(i)} \Delta \vect{X}_{ij}^{(l)} \cdot \textit{CoordMLP}(\vect{H}_{e_{ij}}^{(l+1)}),
\end{equation*}
where $\mathcal{N}_i$ denotes the neighboring nodes of $v_i$ in graph $G$.
Finally, node features are updated by aggregating neighboring edge information with residual connections:
\begin{equation*}
    \vect{H}_i^{(l+1)} = \vect{H}_i^{(l)} + \textit{NodeMLP} \left(\vect{H}_i^{(l)} \oplus \sum_{j \in \mathcal{N}(i)} \vect{H}_{e_{ij}}^{(l+1)} \right)
\end{equation*}
This E(3)-equivariant (refer to Definition~\ref{def:e3-equiv-inv} in Appendix~\ref{appendix-sec:subsec-e3-def}) message passing scheme ensures that both spatial and biochemical properties are properly captured and updated during the message passing process.
\begin{theorem}
\label{thm:thm1-emp}
    For any transformation $T \in E(3)$, we have $\vect{H}_i^{(l+1)}, T(\vect{X}_i^{(l+1)}) = \textit{EMP}\left(\vect{H}_i^{(l)}, T(\vect{X}_i^{(l)})\right)$, where $T(X) := \vect{Q}\vect{X} + \vect{b}$ denotes the E(3) transformation of $\vect{X}$.
\end{theorem}
All proofs of theorems are provided in Appendix~\ref{appendix-sec:proof}.

\subsection{Inter-Graph Refinement}
\label{sec:subsec-inter-graph}
The inter-graph $G_{inter} = (V_p,V_{ae},E_{inter})$ represents interactions in the antibody-antigen binding interface between antibody paratope residues $V_{p}$ and antigen epitope residues $V_{ae}$. 
The key innovation of Igformer lies in its sophisticated inter-graph refinement strategy, which combines personalized propagation with global attention to capture complex antibody-antigen interactions. 
Our approach integrates two complementary components: {\em approximate personalized propagation (APP)} for preserving local structural information and {\em simplified graph transformer (SGFormer)} for capturing global dependencies.
Next, we elaborate on the graph generation and refinement process.

{\bf Approximate Personalized Propagation.}
We construct intra-graphs $G_{ab}$ and $G_{ae}$ by establishing connectivity between $k$-nearest residues based on atomic-level coordinate comparisons, ensuring that graph connectivity reflects spatial relationships within both antibody and antigen components. A detailed description of the intra-graph construction process is provided in Appendix~\ref{appendix-sec:graph-construct}.

Given the residue embeddings $\vect{H}^{(0)}$ updated by EMP in the intra-graph, APP performs iterative message passing while maintaining a balance between local and global information:
\begin{equation}
\label{eq:eq-app}
    \vect{H}^{(j+1)} = (1 - \alpha) \cdot \hat{\vect{P}} \vect{H}^{(j)} + \alpha \cdot \vect{H}^{(0)},
\end{equation}
where $\hat{\vect{P}} = \hat{\mathbf{D}}^{-1/2} \hat{\mathbf{A}} \hat{\mathbf{D}}^{-1/2}$ denotes the normalized transition matrix in the intra-graph with self-loops, and $\alpha$ controls the retention of initial node features. This propagation scheme enables the updated embeddings to capture neighborhood context while preserving crucial local structural information \cite{klicpera2019appnp}.

{\bf Simplified Graph Transformer.}
Subsequently, we introduce a Simplified Graph Transformer (SGFormer) component that implements a global attention mechanism to dynamically model complex interactions between residues across the entire binding interface $V_{\textit{inter}} = V_{p} \cup V_{ae}$.
Specifically, for a pair of residues $(v_i,v_j)$ in the given antibody-antigen binding region, we compute the attention weight:
\begin{equation*}
    a_{ij} = \frac{\exp\left(\vect{H}_i^{\top} \vect{W} \vect{H}_j\right)}{\sum_{k \in V} \exp\left(\vect{H}_i^{\top} \vect{W} \vect{H}_k\right)},
\end{equation*}
where $\vect{W}$ is a learnable weight matrix. Residue embeddings are updated through attention-weighted aggregation:
\begin{equation*}
    \vect{H}_i = \sum_{j \in V} a_{ij} \cdot \vect{H}_j.
\end{equation*}
This self-attention mechanism allows SGFormer to globally capture interactions between all residues in the binding region, thereby capturing both local information and long-range dependencies. 

{\bf Inter-Graph Refinement.}
To model the binding interface dynamics, we encode bi-directional interactions between antibody paratope and antigen epitope residues into edge embeddings.
For each residue pair $(v_i,v_j)$, where $v_i$ belongs to the antibody paratope and $v_j$ to the antigen epitope. The updated pair-wise distance $dist_{ij}$ is:
\begin{equation*}
    dist_{ij} = \textit{EdgeMLP}\left(\vect{H}_i \oplus \vect{H}_j, \vect{H}_j \oplus \vect{H}_i \right).
\end{equation*}
These learned distances capture the complex chemical and spatial relationships between residues, enabling dynamic refinement of the inter-graph structure through $k$-nearest neighbor selection. The complete inter-graph refinement process is detailed in Appendix~\ref{appendix-sec:graph-construct}.
\begin{remark}
    The integration of APP and SGFormer enables comprehensive refinement of both residue and edge representations. APP preserving crucial local structural properties while facilitating efficient information propagation, and the attention mechanism in SGFormer captures dynamic residue interactions across the entire binding interface. This dual-approach enables precise modeling of both local chemical interactions and long-range structural dependencies, which is crucial for accurate antibody structure prediction, as we will show in our experiments.
\end{remark}

\subsection{Iterative Update} 
\label{sec:subsec-update}
Our framework implements a sophisticated iterative learning process that incorporates triangle attention mechanism \cite{jumper2021alphafold} with a dual-EMP module to capture complex interactions within antibody-antigen binding interfaces. The process begins with the initial embeddings $\vect{H}_{full} = \vect{H}_{ae} \cup \vect{H}_{ab}$ and extracts the interface-specific representations $\vect{H}_{inter} = \vect{H}_{ae} \cup \vect{H}_{p}$ from $\vect{H}_{full}$. The corresponding coordinates $\vect{X}_{full} = \vect{X}_{ae} \cup \vect{X}_{ab}$ and $\vect{X}_{inter} = \vect{X}_{ae} \cup \vect{X}_{p}$ are initialized according to Appendix~\ref{appendix-sec:coord}. 
For clarity, we denote $\vect{H}_{full}$ as $\vect{H}_{intra}$ and $\vect{X}_{full}$ as $\vect{X}_{intra}$ in subsequent discussions.

{\bf Triangle Multiplicative Module.}
Within the paratope region $V_p$ containing $n$ residues, we construct a normalized interaction matrix $\vect{Z} \in \mathbb{R}^{n \times n}$ to encode pairwise residue relationships for each residue pair $(v_i,v_j)$:
\begin{equation*}
    \vect{Z}_{ij} = \textit{LayerNorm} \left(\textit{MLP} \left([\vect{H}_i \oplus \vect{H}_j]\right) \right).
\end{equation*}
The concatenation operation $\oplus$ ensures directional sensitivity through its non-commutative nature. This interaction matrix undergoes iterative refinement through two complementary mechanisms: triangle multiplicative module and axial attention.

The triangle multiplicative module processes embeddings through learnable normalized projections to generate pair-wise residue interactions in the paratope:
\begin{equation*}
    \vect{H}_l = \vect{W}_l \vect{Z}, \quad \vect{H}_r = \vect{W}_r \vect{Z}, \quad \vect{I}_{ij} = (\vect{H}_l)_i^{\top} (\vect{H}_r)_j,
\end{equation*}
where $\vect{W}_l$ and $\vect{W}_r$ are learnable weight matrices.
The resulting pair-wise interactions are then modulated by a gating mechanism and projected into the embedding space:
\begin{equation*}
    f_{triangle}^{out}(\vect{H}) = \vect{W}_o \cdot \sigma(\vect{W}_g \vect{Z}) \cdot \vect{I},
\end{equation*}
where $ \sigma(\cdot) $ is the sigmoid function.

{\bf Axial Attention.}
The row-wise axial attention mechanism computes dynamic residue relationships through a scaled dot-product attention:
\begin{equation*}
    f_{att}^{out}(\vect{H}) = \text{softmax}\left( \frac{\vect{Q}^{\top} \vect{K}}{\sqrt{d_k}} \right) \cdot \vect{V}, 
\end{equation*}
where residue embeddings in $\vect{Q} = \vect{W}_q \vect{H}$, $\vect{K} = \vect{W}_k \vect{H}$, and $\vect{V} = \vect{W}_v \vect{H}$ are generated by $f_{triangle}^{out}(\vect{H})$, and $\vect{W}_q$, $\vect{W}_k$, and $\vect{W}_v$ are learnable parameters. The final paratope outgoing representation integrates both attention and triangle multiplicative mechanism:
\begin{equation*}
    f_{out}(\vect{H}) = f_{triangle}^{out}(\vect{H}) + f_{att}^{out}(\vect{H}).
\end{equation*}
The representations of the paratope region are updated iteratively integrating both outgoing and incoming representations:
\begin{equation}
\label{eq:eq-paratope}
    \vect{H}^{(k+1)} = \vect{H}^{(k)} + f_{out}(\vect{H}^{(k)}) + f_{in}(\vect{H}^{(k)}),
\end{equation}
where $f_{in}(\vect{H})$ represents the ingoing update analogous to the outgoing formulation but adopts ingoing (column-wise) multiplicative interactions and column-wise attention.
This bidirectional updating process ensures that paratope embedding $\vect{H}$ captures both pairwise geometric interactions via $ f_{triangle} $ and dynamic global dependencies via $f_{att}$.

After $Ks$ iterations, the diagonal elements of paratope embeddings calculated by Equation~\ref{eq:eq-paratope} are extracted to generate the final embeddings of paratope:  
\begin{equation}
\label{eq:eq-paratope-final}
    \vect{H}_{p} = \textit{MLP} \left(\textit{diag} \left( \vect{H}^{(K)} \right) \right),
\end{equation}

{\bf Dual EMP Module.} Next, we employ a dual-scale message-passing framework using two separate EMP modules to model intra-graph $G_{intra}$ and inter-graph $G_{inter}$ interactions:
\begin{equation}
\label{eq:eq-dual-emp}
\begin{aligned}
    \left( \vect{H}_{\textit{intra}}^{(t+1)}, \vect{X}_{\textit{intra}}^{(t+1)}\right) &= \textit{EMP}_{\textit{intra}} \left(\vect{H}_{\textit{intra}}^{(t)}, \vect{X}_{\textit{intra}}^{(t)} \right), \\
    \left( \vect{H}_{\textit{inter}}^{(t+1)}, \vect{X}_{\textit{inter}}^{(t+1)}\right) &= \textit{EMP}_{\textit{inter}} \left( \vect{H}_{\textit{inter}}^{(t)}, \vect{X}_{\textit{inter}}^{(t)} \right).
\end{aligned}
\end{equation}
At each iteration, we update our graph representations according to Equation~\ref{eq:eq-dual-emp}, and substitute the paratope embeddings in the intra-graph $\vect{H}_{\textit{intra}}^{(t+1)}$ using $\vect{H}_{\textit{inter}}^{(t+1)}$. On the other hand, coordinates for intra- and inter-graph are maintained separately. This dual-scale message-passing framework enables comprehensive modeling of both local residue relationships and global antigen-antibody interactions, with each scale optimized for its specific context.
Finally, residue connections are incorporated in the updated embeddings for stable learning across different training iterations.

\subsection{Igformer Pipeline}
\label{sec:subsec-pipeline}
We present the key stages of the Igformer learning pipeline, with detailed training algorithms provided in Appendix~\ref{appendix-sec:pipeline}. 
The process begins with the initialization of residue representations and coordinates. Following this, the intra-graph is constructed for personalized information propagation, which is then processed by the SGFormer for inter-graph refinement. The embeddings generated in this process serve exclusively for inter-graph refinement and are not involved in the representation learning process. Concurrently, a triangle attention module updates the representation of each residue in the paratope. Finally, a dual EMP module processes residue representations to generate the final coordinates and representations, which are then used for downstream structure and sequence prediction tasks.

The following theorem indicates that the coordinates generated by Igformer are E(3)-equivariant and the residue embeddings are invariant.
\begin{theorem}
\label{thm:thm-igformer}
    Let $\hat{\vect{H}}_i, \hat{\vect{X}}_i = \textit{Igformer} \left( \vect{H}_i^{(0)}, \vect{X_i} \right)$ denote the embedding and coordinate of $v_i$ generated by Igformer.
    For any transformation $T \in E(3)$, we have $\hat{\vect{H}}_i, T(\hat{\vect{X}}_i) = \textit{Igformer}\left(\vect{H}_i^{(0)}, T(\vect{X_i})\right)$, where $T(X) := \vect{Q}\vect{X} + \vect{b}$ denotes the E(3) transformation of $\vect{X}$.
\end{theorem}

\subsection{Prediction and Loss Function}
\label{sec:subsec-loss}

{\bf Sequence Prediction.}
The $i$-th position in the amino acid sequence of the paratope is predicted using its embedding:
$$
    s_i = \textit{softmax}\left( \textit{MLP}\left(\vect{H}_i\right) \right)
$$
{\bf Structure Prediction.}
The 3D structures of paratope are generated by the coordinates $\vect{X}$ computed in the last round of Dual EMP Module via Equation~\ref{eq:eq-dual-emp}.

{\bf Loss Function.}
The loss function of Igformer consists of three components:
\begin{equation}
\label{eq:eq-loss}
  \mathcal{L} =  \mathcal{L}_{\textit{seq}} + \mathcal{L}_{\textit{struct}} +  \mathcal{L}_{\textit{interface}}
\end{equation}
Here, $\mathcal{L}_{\textit{seq}}$ is the cross entropy loss that minimizes the dissimilarities between predicted and original 1D sequences. $\mathcal{L}_{\textit{struct}}$ minimizes the difference between reconstructed and ground-truth 3D structures. $\mathcal{L}_{\textit{interface}}$ optimizes the reconstructed inter-graph. Detailed descriptions of each component are provided in Appendix~\ref{appendix-sec:loss}.

\section{Supplementary}
\label{appendix-sec:sup}
\subsection{Graph Construction}
\label{appendix-sec:graph-construct}
{\bf Intra-Graph Construction. } The intra-graph construction process integrates both sequence and structural information from the antibody-antigen complex. We begin by extracting information from the antigen epitope $G_{ae}$ and antibody $G_{ab}$, where each residue $i$ serves as a graph node $v_i$. Each node encapsulates both its amino acid type and 3D coordinate information, providing a comprehensive representation of its chemical and spatial properties.
We formalize the intra-graph as $G_{\textit{Intra}} = (V_{\textit{Intra}}, E_{\textit{Intra}})$, where $V_{\textit{Intra}} = V_{ae} \cup V_{ab}$ represents the union of antigen epitope residues $V_{ae}$ and antibody residues $V_{ab}$. Edge construction follows a distance-based approach, connecting residues within their respective chains (either epitope or antibody). For each residue $v_i$, we establish connections to its $k$-nearest neighbors $v_j$ within the same chain based on Euclidean distances between their spatial coordinates. The shortest distance between two residues is defined as the minimum distance between any pair of atoms across two residues.
This construction ensures that the intra-graph captures meaningful spatial relationships within both the antigen epitope and antibody domains while maintaining their distinct molecular contexts.

{\bf Inter-Graph Construction}
The inter-graph construction process integrates both sequence and structural information from the antibody-antigen binding region.
We construct an inter-graph $G_{\textit{Inter}} = (V_{p}, V_{ae}, E_{\textit{Inter}})$ to model critical paratope-epitope interactions between antibody paratope residues $V_p$ and antigen epitope residues $V_{ae}$.
Edge construction in the inter-graph is based on embedding similarity between residue pairs. For paratope residues $v_i \in V_p$ and $v_j \in V_{ae}$ with embeddings $\vect{H}_i$ and $\vect{H}_j$, we compute bidirectional interaction features:
\begin{equation*}
    \vect{z}_{\textit{in}} = \vect{H}_i \oplus \vect{H}_j, \quad \vect{z}_{\textit{out}} = \vect{H}_j \oplus \vect{H}_i.
\end{equation*}
These concatenated embeddings are processed through a feed-forward neural network (FFN) to generate interaction scores:
\begin{equation}
    \textit{FFN}(\vect{z}) = \sigma(\vect{W}_2 \cdot \sigma(\vect{W}_1 \cdot \vect{z} + \vect{b}_1) + \vect{b}_2),
\end{equation}
where $ \vect{W}_1 \in \mathbb{R}^{d \times 2d}$, $ \vect{W}_2 \in \mathbb{R}^{d \times d}$, and $ \vect{b}_1, \vect{b}_2 \in \mathbb{R}^d$ are learnable parameters, $ \sigma$ is the activation function. The final pair-wise distance of $(v_i,v_j)$ combines interactions from both directions:
\begin{equation*}
    dist_{ij} = \textit{FFN}(\vect{z}_{\textit{in}}) + \textit{FFN}(\vect{z}_{\textit{out}}).
\end{equation*}
The edge set $E_{\textit{inter}}$ is then constructed by selecting $k$-nearest residues of each $v_i \in V_p$ based on these interaction scores.
This construction focuses exclusively on paratope-epitope interactions, providing a refined inter-graph structure of the binding interface that captures the essential dynamics of antibody-antigen docking.

\subsection{Residue Embedding}
\label{appendix-sec:embeds}

Our embedding framework captures comprehensive residue characteristics through two complementary components: residue-type and positional embeddings.
First, we encode residue-type information through $\vect{H}_i^{\textit{res}} \in \mathbb{R}^d$, which represents the biochemical properties of amino acid residue $v_i$.
This embedding resides in a matrix of dimensions $N_{\textit{res}} \times d$, where $N_{\textit{res}} = 20$ represents the number of different amino acid types.

To preserve structural context, we implement position-specific embeddings 
$\vect{H}_i^{\textit{pos}} \in \mathbb{R}^{d}$ that encode the sequential context of residues. This positional encoding maintains distinct indices for antigen, heavy chain, and light chain segments, enabling the model to differentiate between different components.
The position embedding matrix has dimensions $N_{\textit{seq}} \times d_{\textit{pos}}$, where $N_{\textit{seq}} = 192$ is the maximum length of the sequence, ensuring consistent positional representation across all residue types while maintaining chain-specific contexts.

Importantly, different segments of the sequence, such as the antigen, heavy chain, and light chain, utilize distinct position indices to allow the model to differentiate between these regions.

The final residue representation integrates both biochemical and positional information through a simple addition:
\begin{equation}
    \vect{H}_i = \vect{H}_i^{\textit{res}} + \vect{H}_i^{\textit{pos}}.
\end{equation}

\subsection{Template-based Coordinates}
\label{appendix-sec:coord}

\subsubsection{Template Generation}
\label{appendix-sec:subsec-template-coord}
The template generation process exploits a fundamental characteristic of antibody structures: the high conservation within framework regions (FRs). 
We define a residue as well-conserved when its amino acid type is maintained across at least 95\% of antibodies in the dataset - a threshold optimized through empirical analysis to balance the number of conserved residues against their structural variability \cite{kong2023dymean}. Using this criterion, we identify 16 and 18 well-conserved residues in the heavy and light chains, respectively. To construct template coordinates $\vect{X}^{\textit{temp}}$, we compute mean backbone atom coordinates (N, CA, C, O) in the main chains across all antibodies in the dataset, establishing an initial backbone template $\{\vect{X}_r \in \mathbb{R}^{3 \times 4} | r \in \mathcal{W}_{\textit{temp}} \}$, where $\mathcal{W}_{\textit{temp}}$ denotes the positions of well-conserved residues as detailed in Table \ref{tab:template_positions}.
\begin{table}[t]
\small
\centering
\vspace{-2mm}
\caption{Well-conserved residue positions used in template generation.}
\begin{tabular}{c|l|c}
\hline \hline
Chain & Positions & Count \\
\hline
Heavy (H) & 8, 15, 23, 41, 44, 50, 52, 98, 100, 102, 104, 118, 119, 121, 126, 127 & 16 \\
\hline
Light (L) & 5, 6, 16, 23, 41, 70, 75, 76, 79, 89, 91, 98, 102, 104, 118, 119, 121, 122 & 18 \\
\hline \hline
\end{tabular}
\label{tab:template_positions}
\end{table}

\subsubsection{Coordinate Initialization}
\label{appendix-sec:subsec-coord-init}
Our coordinate initialization strategy employs a template-based approach that preserves actual coordinates for the epitope region $V_e$ while utilizing template-derived coordinates for paratope regions requiring prediction.

Each residue is represented using 14 atoms in 3D space, with coordinates denoted as $\vect{X}_{\textit{ae}} = \{ \vect{X}_i \in \mathbb{R}^3 \mid v_i \in V_{ae} \}$ for known positions and $\vect{X}^{\textit{temp}} = \{ \vect{X}_r \in \mathbb{R}^3 \mid j \in \mathcal{W}_{\textit{temp}} \}$ for template-derived positions (refer to Appendix~\ref{appendix-sec:subsec-template-coord}).

For residues with missing coordinates in the template, linear interpolation is employed. The missing coordinate $\vect{X}_k$ for position $k$ is calculated as follows:
\begin{equation*}
    \vect{X}_k = \vect{X}_{\textit{left}} + \frac{k - \textit{left}}{\textit{right} - \textit{left}} (\vect{X}_{\textit{right}} - \vect{X}_{\textit{left}}),
\end{equation*}
where $\textit{left}$ and $\textit{right}$ denote indices of nearest known template coordinates. The intersection of these indices is denoted as $\mathcal{W}_{\textit{align}} = \mathcal{W}_{\textit{true}} \cap \mathcal{W}_{\textit{temp}}$.

To ensure interpolation remains within valid sequence boundaries, $\textit{left}$ and $\textit{right}$ are constrained as follows:  $0 \leq \textit{left} < k < \textit{right} \leq N_{\textit{seq}},$
where $ N_{\textit{seq}} $ is the total length of the original antibody sequence.   
This ensures that interpolation remains within the structural constraints of the antibody sequence while preserving the geometric continuity of the coordinates.  

To ensure structural coherence, we align template coordinates with true coordinates by optimizing a rigid-body transformation for shared indices $\mathcal{W}_{\textit{align}}$. The optimal transformation, comprising rotation matrix $\vect{Q} \in \mathbb{R}^{3 \times 3}$ and translation vector $\vect{t} \in \mathbb{R}^3$, is determined by solving the following optimization problem:
\begin{equation}
    \arg \min_{\vect{Q}, \vect{t}} \sum_{i \in \mathcal{W}_{\textit{Align}}, j \in \mathcal{W}_{\textit{temp}}} \left| \vect{X}_i - (\vect{Q} \vect{X}_j + \vect{t}) \right|^2,
\end{equation}
subject to $\vect{Q}^\top \vect{Q} = \vect{I}$. After interpolation, we obtain the complete antibody coordinates as  
\begin{equation}
    \vect{X}_{ab}^{'} = \vect{X}^{\textit{temp}} \cup \vect{X}^{\textit{interp}},
\end{equation}
where $\vect{X}^{\textit{interp}}$ represents the set of interpolated coordinates for positions that are within the sequence range but absent from the template. Formally,  
\begin{equation}
    \vect{X}^{\textit{interp}} = \{ \vect{X}_k \mid k = \{0, \cdots, N_{\textit{seq}}-1 \} \text{ and } k \notin \mathcal{W}_{\textit{temp}} \}.
\end{equation}
This ensures that interpolation is applied only to residues that are part of the original antibody sequence but do not have known coordinates in the template. The resulting transformation for alignment is then applied globally:
\begin{equation}
    \vect{X}_{ab}  = \vect{Q} \vect{X}_{ab}^{'} + \vect{t}, 
\end{equation}
This initialization procedure primarily focuses on backbone atoms (N, CA, C, O), with side-chain atoms initially positioned at their respective CA coordinates, establishing a foundation for subsequent refinement.

\subsubsection{Coordinate Processing}
\label{appendix-sec:subsec-coord-process}

Following initial coordinate assignment, we implement a comprehensive normalization procedure to ensure consistent scale and reference frame across the antibody-antigen complex $\vect{X}$. For a system with $N$ total atoms, we first define an indicator function $\mathbb{I}(i)$ that distinguishes between antigen and antibody atoms:
\begin{equation} 
\label{eq:coord-indicator}
    \mathbb{I}(i) = 
    \begin{cases}
        1 & \textit{if }v_i \in V_\textit{ae} \\
        0 & \textit{if }v_i \in V_\textit{ab}
    \end{cases}.
\end{equation} 
Using this indicator function, we compute chain-specific centers of mass:
\begin{equation} 
\label{eq:coord-center}
    \vect{X}_c = \frac{\sum_{i=1}^N \vect{X}_i \cdot \mathbb{I}(i)}{\sum_{i=1}^N \mathbb{I}(i)}.
\end{equation}
The coordinates are then centered relative to their respective chain centers $\vect{X}_i^{'} = \vect{X}_i -  \vect{X}_c$.
Finally, we apply dimension-specific normalization to ensure consistent scale:
\begin{equation*}
    \vect{X}_i = \frac{\vect{X}_i^{'} - \mu_d}{\sigma_d},
\end{equation*}
where $\mu_d$ and $\sigma_d$ represent the mean and standard deviation along dimension $d$ across all atoms, with a target standard deviation of $\sigma=10$. This processing pipeline ensures that all atomic coordinates are properly scaled and centered within their respective molecular contexts, facilitating subsequent structural prediction tasks. Finally, we update the global coordinates again based on Equations \ref{eq:coord-indicator} and \ref{eq:coord-center}.

\subsubsection{Paratope Coordinates Generation}
\label{appendix-sec:subsec-coord-generate}
Our paratope coordinate generation framework implements a structured approach to initialize and refine binding interface positions. The process begins by anchoring initial paratope coordinates to the epitope center $\vect{X}_0 = \vect{X}_c$
where $\vect{X}_0$ represents the initial paratope coordinates and $\vect{X}_c$ represents the coordinates of the epitope center calculated by Equation \ref{eq:coord-center}, which are assigned as initial positions for all atoms in the paratope region.
We then introduce controlled structural variations:
\begin{equation*}
    \vect{\epsilon}_{i} = 
    \begin{cases} \vect{\epsilon}_{ca} & \textit{for CA atoms} \\
    0.1\vect{\epsilon}_{o} + \vect{\epsilon}_{ca} & \textit{for other atoms}
\end{cases},
\end{equation*}
where 
$\vect{\epsilon}_{o}, \vect{\epsilon}_{ca} \sim  \mathcal{N}(0, 1)$. This differential noise incorporation ensures that CA atoms exhibit greater conformational flexibility while maintaining structural coherence for other atoms through correlated, smaller perturbations. 
The final paratope coordinates are computed as:
\begin{equation*}
    \vect{X}_{\textit{p}} = \vect{X}_0 + \vect{\epsilon}_{i}.
\end{equation*}

Our model maintains two distinct coordinate representations:  
\begin{equation}
    \vect{X}_{\textit{full}} = \vect{X}_{\textit{ae}} \cup \vect{X}_{\textit{ab}},
\end{equation}
where $\vect{X}_{\textit{full}}$ represents the complete antibody-antigen complex, with  
$\vect{X}_{\textit{ae}}$ denoting the epitope coordinates and $\vect{X}_{\textit{ab}}$ representing the full antibody structure, including framework and paratope regions.  

Similarly, we define the focused binding interface representation as:  
\begin{equation}
    \vect{X}_{\textit{inter}} = \vect{X}_{\textit{ae}} \cup \vect{X}_{\textit{p}},
\end{equation}
where $\vect{X}_{\textit{inter}}$ includes the epitope coordinates $\vect{X}_{\textit{ae}}$ and the generated paratope coordinates $\vect{X}_{\textit{p}}$, without the full antibody framework.  

This distinction ensures that the model can separately handle global structural representations $\vect{X}_{\textit{full}}$ and the more localized binding interactions $\vect{X}_{\textit{inter}}$ for improved prediction and analysis.

We employ the standardized IMGT/Chothia \cite{chothia1989cdr,lefranc2003imgt} to maintain structural consistency across antibody chains.This well-established system assigns independent position indices to each chain while preserving the original residue numbering from PDB structures. For the antibody heavy chain, the framework regions and CDRs follow specific positional ranges, with CDR-H1 typically spanning positions 23-35 and CDR-H3 occupying positions 104-118. Similarly, the light chain maintains its distinct numbering convention, with CDR-L1 typically located at positions 27-38 and CDR-L3 at positions 105-117. This independent numbering approach ensures clear differentiation between chain segments while facilitating accurate structural analysis and prediction. To maintain a clear distinction between antibody and antigen components, all antigen residues are uniformly assigned position 0, allowing the model to effectively distinguish between interacting molecular components.

\subsection{Epitope Selection}
\label{appendix-sec:epitope}
We implement a distance-based approach for epitope selection that identifies the most relevant antigen residues interacting with the antibody CDR-H3. For each antigen residue $v_i \in V_{ag}$, we compute its minimal distance to any CDR-H3 residue $v_j$:
$$
    \hat{dist}(v_i) = \min_{v_j} dist_{i, j},
$$
From these distances, we select the residues with top-$k$ (typically $k=48$) smallest distances to form the antigen epitope:
\begin{equation*}
    V_{\textit{ae}} = \{ \arg\min_k \hat{dist}(v_i) \mid \forall v_i \in V_{ag} \}.
\end{equation*}
This selection process ensures that we focus on the most critical residues involved in antibody-antigen binding while maintaining computational efficiency. The fixed-size antigen epitope selection provides a consistent representation of the binding interface, enabling robust structural prediction and analysis.

\subsection{Proofs}
\label{appendix-sec:proof}
In this section, we present the proofs of theorems in this paper, beginning with formal definitions of SE(3) equivariance and SE(3) invariance.

\subsubsection{Preliminaries}
\label{appendix-sec:subsec-e3-def}
\begin{definition}[E(3)-equivariance and SE(3)-invariance]
\label{def:e3-equiv-inv}
Let $\varphi : \mathcal{X} \to \mathcal{Y}$ be a mapping function, and let $T \in E(3)$ denote a rigid transformation in three-dimensional space comprising rotation and reflection, $\vect{Q} \in O(3)$ and a translation $\vect{b}\in\mathbb{R}^3$. The transformation acts on coordinates as:
$$
   T(\vect{X}) := \vect{Q}\vect{X} + \vect{b}, 
   \quad \vect{Q} \in O(3), \vect{b} \in \mathbb{R}^3.
$$
Then:
\begin{itemize}[topsep=0.5mm, partopsep=0pt, itemsep=0pt, leftmargin=10pt]
    \item \(\varphi\) is \textbf{E(3)-equivariant}, if for any transformation $\vect{X}\mapsto \vect{X}_{\text{out}} = \varphi(\vect{X})$, we have $T(\vect{X}_{\text{out}}) = \varphi\bigl(T(\vect{X})\bigr)$.
    \item \(\varphi\) is \textbf{E(3)-invariant} if $\varphi \bigl(T ( \vect{X})\bigr) = \varphi(\vect{X})$ for all \(T\in \mathrm{E}(3)\), i.e. $T$ acts as identity transformation in the output space.
\end{itemize}
\end{definition}

In our model architecture, 3D coordinates are E(3)-equivariant, meaning they transform consistently with input transformations, while residue embeddings demonstrate E(3)-invariance, remaining unchanged under rigid transformations. 

{\bf Notations.} Next, we summarize the notations used throughout our analysis.
\begin{itemize}[topsep=0.5mm, partopsep=0pt, itemsep=0pt, leftmargin=10pt]
    \item Let $\{(\vect{H}_i^{(l)}, \vect{X}_i^{(l)})\}_{i=1}^n$ denote the residue embeddings \(\vect{H}_i^{(l)} \in \mathbb{R}^d\) and 3D coordinates $\vect{X}_i^{(l)} \in \mathbb{R}^3$ at layer $l$.  
    \item Let $\vect{H}_{e_{ij}}^{(l)}$ be the edge feature between residues $v_i$ and $v_j$ at layer $l$.  
    \item Let $T(\vect{X}) := \vect{Q}\vect{X} + \vect{b}$ denote a rigid transformation in $E(3)$.
\end{itemize}

Below, we will show that Igformer maintains coordinate equivariance and embedding invariance under these transformations at each layer, following Definition~\ref{def:e3-equiv-inv}.

\subsubsection{E(3)-Equivariance of the Coordinates}

We begin by establishing the E(3)-equivariance property of coordinate predictions in the EMP module.  

\begin{lemma}[Linear Transformation of Coordinate Differences]
\label{lemma:deltaX}
    For coordinates differences \(\Delta \vect{X}_{ij}^{(l)} := \vect{X}_i^{(l)} - \vect{X}_j^{(l)}\), applying \(T\in E(3)\) to coordinates:
    $
       \widetilde{\vect{X}}_i^{(l)} = T(\vect{X}_i^{(l)}) 
       = \vect{Q}\,\vect{X}_i^{(l)} + \vect{b},
    $
    results in  
    $
       \Delta \widetilde{\vect{X}}_{ij}^{(l)} = \vect{Q}\,\Delta \vect{X}_{ij}^{(l)}.
    $
\end{lemma}

\begin{proof}
By direct substitution:
\[
  \widetilde{\vect{X}}_i^{(l)} - \widetilde{\vect{X}}_j^{(l)}
  = (\vect{Q}\vect{X}_i^{(l)} + \vect{b}) - (\vect{Q}\vect{X}_j^{(l)} + \vect{b})
  = \vect{Q}\bigl(\vect{X}_i^{(l)} - \vect{X}_j^{(l)}\bigr) = \vect{Q}\Delta \vect{X}_{ij}^{(l)}.
\]
The translation $\vect{b}$ cancels, leaving only the rotation $\vect{Q}$.
\end{proof}

\begin{lemma}[Geometry-Based Operations]
\label{lemma:geom-invariant}
    Geometric Operations in Igformer, including pairwise distances $\|\vect{X}_i - \vect{X}_j\|$, dot products, and dihedral angles, remain invariant under any $T \in E(3)$.
\end{lemma}

\begin{proof}
   A translation $\vect{b}$ always cancels in $\vect{X}_i - \vect{X}_j$.  A rotation $\vect{Q}\in O(3)$ preserves Euclidean norms and dot products:$ \|\vect{Q}\vect{v}\|^2 = \|\vect{v}\|^2$. Therefore, all geometric operations remain invariant under $T$.
\end{proof}

\begin{theorem}[E(3)-Equivariance of Coordinates in EMP]
\label{thm:coord-equivar}
    At layer \(l\), Igformer updates coordinates through the following sequential operations:
    \begin{align*}
      \vect{S}_{ij}^{(l)} 
       &= \textit{MLP}\bigl(\|\Delta \vect{X}_{ij}^{(l)}\|^2,\,\dots\bigr), \quad \text{(geometry-based score)} \\
      \vect{H}_{e_{ij}}^{(l+1)} 
       &= \textit{EdgeMLP} \left(\vect{H}_i^{(l)} \oplus \vect{H}_j^{(l)} \oplus \vect{S}_{ij}^{(l)} \oplus \vect{H}_{e_{ij}}^{(l)} \right), \\
      \vect{X}_i^{(l+1)}
       &= \vect{X}_i^{(l)} + \sum_{j\in \mathcal{N}(i)}
            \Delta \vect{X}_{ij}^{(l)} \cdot
            \textit{CoordMLP}\bigl(\vect{H}_{e_{ij}}^{(l+1)}\bigr).
    \end{align*}
    This update rule is E(3)-equivariant on coordinates. Concretely, applying $T$ to input coordinates $\{\vect{X}_i^{(l)}\}$ results in the same transformation $T$ being applied to output coordinates $\{\vect{X}_i^{(l+1)}\}$.
\end{theorem}

\begin{proof}
By Lemma~\ref{lemma:geom-invariant}, $\vect{S}_{ij}^{(l)}$ and consequently $\vect{H}_{e_{ij}}^{(l+1)}$ remain unchanged when input coordinates are transformed by $T$. 
Let
$
   \vect{w}_{ij}^{(l+1)} := \textit{CoordMLP}\bigl(\vect{H}_{e_{ij}}^{(l+1)}\bigr).
$
Then:
$$
   \vect{X}_i^{(l+1)} = \vect{X}_i^{(l)} + \sum_{j\in \mathcal{N}(i)}
   \Delta \vect{X}_{ij}^{(l)} \cdot \vect{w}_{ij}^{(l+1)}.
$$
Under transformation $T: \vect{X}_i^{(l)}\mapsto \vect{Q}\vect{X}_i^{(l)}+\vect{b}$, we have:
$$
\begin{aligned}
   \widetilde{\vect{X}}_i^{(l+1)}
   &= \widetilde{\vect{X}}_i^{(l)}
      + \sum_{j\in\mathcal{N}(i)}
        \bigl[\widetilde{\vect{X}}_i^{(l)} - \widetilde{\vect{X}}_j^{(l)}\bigr]
        \cdot \vect{w}_{ij}^{(l+1)}
   \quad &\bigl(\text{same } \vect{w}_{ij}^{(l+1)}\bigr) 
   \\[4pt]
   &= (\vect{Q}\vect{X}_i^{(l)}+\vect{b})
      +
      \sum_{j\in\mathcal{N}(i)}
        \Bigl[\vect{Q}(\vect{X}_i^{(l)}-\vect{X}_j^{(l)})\Bigr]
        \cdot \vect{w}_{ij}^{(l+1)}
   &\bigl(\text{Lemma~\ref{lemma:deltaX}}\bigr)
   \\[4pt]
   &= \vect{Q}\,\vect{X}_i^{(l)} + \vect{b}
      + \vect{Q}\sum_{j\in\mathcal{N}(i)}
          (\vect{X}_i^{(l)} - \vect{X}_j^{(l)}) \cdot \vect{w}_{ij}^{(l+1)}
   \\[4pt]
   &= \vect{Q} \Bigl[\vect{X}_i^{(l)}
       + \sum_{j\in\mathcal{N}(i)}
         \Delta \vect{X}_{ij}^{(l)} \cdot \vect{w}_{ij}^{(l+1)}\Bigr]
       + \vect{b}
   \\[2pt]
   &= \vect{Q}\vect{X}_i^{(l+1)} + \vect{b} = T\bigl(\vect{X}_i^{(l+1)}\bigr).
\end{aligned}
$$
Therefore, applying $g$ to input coordinates leads to $T\bigl(\vect{X}_i^{(l+1)}\bigr)$ in output coordinates. By induction over all layers, Igformer maintains E(3)-equivariance through its coordinate computation in the EMP module.
\end{proof}

\subsubsection{E(3)-Invariance of the Residue Embeddings}

We now prove that residue embeddings $\vect{H}_i$ remain numerically invariant under global rigid transformation of the input 3D structure.

\begin{lemma}[Base Case for Embedding Invariance]
\label{lemma:base-embed}
    Residue embeddings \(\vect{H}_i^{(0)}\) are initialized using rigid independent features like residue identity, positional index, etc., making them inherently invariant under any transformation $T\in \mathrm{E}(3)$.
\end{lemma}

\begin{lemma}[Edge Feature Update Invariance]
\label{lemma:edge-MLP}
For edge updates of the form:
$$
   \vect{H}_{e_{ij}}^{(l+1)} =
   \textit{EdgeMLP}\Bigl(
      \vect{H}_i^{(l)} \oplus
      \vect{H}_j^{(l)} \oplus
      S_{ij}^{(l)} \oplus
      \vect{H}_{e_{ij}}^{(l)}
   \Bigr),
$$
where $\vect{S}_{ij}^{(l)}$ represents geometry operations (distance, dot product, etc.). If \(\vect{H}_i^{(l)}, \vect{H}_j^{(l)}\) are invariant and $\vect{S}_{ij}^{(l)}$ is unaffected by $T$ (Lemma~\ref{lemma:geom-invariant}), then \(\vect{H}_{e_{ij}}^{(l+1)}\) maintains invariance.
\end{lemma}

\begin{proof}
    The \textit{EdgeMLP} receives numerically identical inputs regardless of coordinate transformation $\vect{Q}\vect{X}_i^{(l)}+\vect{b}$, ensuring $\vect{H}_{e_{ij}}^{(l+1)}$ remains unchanged.
\end{proof}

\begin{theorem}[E(3)-Invariance of Residue Embeddings]
\label{thm:embed-inv}
    Given Igformer's residue embedding update rule:
    $$
      \vect{H}_i^{(l+1)} = \vect{H}_i^{(l)} + 
      \textit{NodeMLP}\Bigl(
        \vect{H}_i^{(l)} \oplus
        \sum_{j\in \mathcal{N}(i)} \vect{H}_{e_{ij}}^{(l+1)}
      \Bigr),
    $$
    or its self-attention variant, the embeddings $\vect{H}_i^{(l+1)}$ remain invariant under any global transformation $T\in \mathrm{E}(3)$. By extension, Igformer's final embeddings $\vect{H}_i$ maintain E(3)-invariance.
\end{theorem}

\begin{proof}
    We prove this theorem through induction:
    \textbf{Base Case.}  By Lemma~\ref{lemma:base-embed}, initial embeddings $\vect{H}_i^{(0)}$ are independent of 3D coordinates.  
    
    \textbf{Inductive Step.}  Assume $\{\vect{H}_i^{(l)}\}$ is already invariant under $T$. 
    \begin{itemize}[topsep=0.5mm, partopsep=0pt, itemsep=0pt, leftmargin=10pt]
        \item According to Lemma~\ref{lemma:geom-invariant}, any geometry-based scalar $\vect{S}_{ij}^{(l)}$ remain unchanged under $T$. 
        \item According to Lemma~\ref{lemma:edge-MLP}, edge features $\vect{H}_{e_{ij}}^{(l+1)}$ maintain invariance under $T$.
        \item Consequently, all inputs to the residue-level update MLP/attention remain unchanged under $T$. 
    \end{itemize}

    Therefore, $\vect{H}_i^{(l+1)}$ maintains invariance under $T$. By induction across layers, Igformer's residue embeddings are E(3)-invariant.
\end{proof}

{\bf Proof of Theorem~\ref{thm:thm1-emp}.}
Combining Theorems~\ref{thm:coord-equivar} and \ref{thm:embed-inv}, we show that coordinates in the EMP module are E(3)-equivariant and embeddings are E(3)-invariant, which finishes the proof of Theorem~\ref{thm:thm1-emp}.

\begin{lemma}
\label{lemma:igr-tmm}
    The embeddings updated by inter-graph refinement and triangle multiplicative module (TMM) modules are invariant.
\end{lemma}

\begin{proof}
    The APP and SGFormer components update embeddings through weighted aggregation:
    $\vect{H}_{i} = \sum_{j \in V}\vect{H}_j$ with residue connections. These operations are independent of coordinates $\vect{X}_i$. Given that $\vect{H}_i$ is invariant under transformation $T$, the MLP/attention operations preserve E(3)-invariance during embedding updates.

    The Triangle Multiplicative module processes pairwise embeddings $\vect{Z}_{ij} = \textit{MLP}(\vect{H}_i,\vect{H}_j)$ without direct dependence on coordinates $\vect{X}_i$. Since $\vect{H}_i$ maintains invariance under $T$, the pairwise features and subsequent internal operations (row/column gating/attention) on $\vect{Z}$ preserves E(3)-invariance. The merging of $\vect{Z}$ into ${\vect{H}_i}$ through MLP/attention operations maintains this invariance property.
\end{proof}

\subsubsection{Proof of Theorem~\ref{thm:thm-igformer}}

Combining Theorem~\ref{thm:coord-equivar}, Theorem~\ref{thm:embed-inv}, and Lemma~\ref{lemma:igr-tmm}, we conclude:

\begin{theorem}[Igformer is E(3)-Equivariant (Coordinates) and E(3)-Invariant (Embeddings)]
\label{thm:igformer-overall}
Let $\{\vect{H}_i^{(0)},\,\vect{X}_i^{(0)}\}_{i\in V}$ be the initial inputs (node features, coordinates). Suppose Igformer generetes final outputs as follows:
$$
  \bigl\{\vect{H}_i^{(\text{final})},\vect{X}_i^{(\text{final})}\bigr\}
  =
  \textit{Igformer} \Bigl(\bigl\{\vect{H}_i^{(0)},\vect{X}_i^{(0)}\bigr\}\Bigr).
$$
Then, for any $T \in \mathrm{E}(3)$, we have:
$$
  \bigl\{\vect{H}_i^{(\text{final})}, T\bigl(\vect{X}_i^{(\text{final})} \bigr) \bigr\}
  = \textit{Igformer}\Bigl( \bigl \{\vect{H}_i^{(0)}, T\bigl(\vect{X}_i^{(0)} \bigr) \bigr\} \Bigr).
$$
\end{theorem}

\begin{proof}
    The proof follows by chaining layer-wise properties:  
    \begin{itemize}[topsep=0.5mm, partopsep=0pt, itemsep=0pt, leftmargin=10pt]
        \item By Theorem~\ref{thm:coord-equivar}, the coordinate update at each layer maintains E(3)-equivariance. Through induction across layers, this property extends to the entire coordinate transformations.
        \item Similarly, Theorem~\ref{thm:embed-inv} and Lemma~\ref{lemma:igr-tmm} establish that residue feature update at each layer preserves E(3)-invariance. By induction, this invariance property carries through to the final residue embeddings. 
    \end{itemize}
    Therefore, the complete Igformer model satisfies both equivariance and invariance properties as stated in the theorem.
\end{proof}

\subsection{Algorithms}
\label{appendix-sec:pipeline}
The training process iterates three times per batch, with each iteration updating the epitope coordinates and embeddings, followed by the reconstruction of intra-graph and inter-graph connections. The final iteration is utilized for loss function calculation.
Algorithm \ref{alg:iterative_update} shows the pseudo-code of the iterative updating process in Igformer. Detailed descriptions from Lines 3-8 are provided in Sections \ref{sec:subsec-inter-graph}-\ref{sec:subsec-update}.
In the first iteration ($T=1$), we use the initialized embeddings for calculation.
For subsequent iterations ($T\in \{2,3\}$), both embeddings and coordinates are updated accordingly. 
(i) The full embedding update is given by $\vect{H}_{full} = \textit{MLP}(\vect{H}_{intra}) + \vect{H}_{full}^{'}$, where $\vect{H}_{full}^{'} = \{ \vect{H}_{ae}, \vect{H}_{ab}^{'} \}$.
The paratope region is specifically updated using attention weights: $\vect{H}_{ab}^{'} = \{\vect{P} \odot \vect{H}_{\textit{p}}, \vect{H}_{ab\setminus\textit{p}}\}$, where $\vect{P} = \textit{MLP}(\vect{H}_{\textit{p}})$ assigns different weights to different residues through element-wise multiplication $\odot$, while maintaining embeddings of non-paratope regions $\vect{H}_{ab\setminus\textit{p}}$. 
(ii) The coordinate update follows a similar strategy, replacing the original epitope coordinates:
$\vect{X}_{full} = \{\vect{X}_{ae}^{'}, \vect{X}_{ab}\}$, ensuring consistent transformation patterns between embedding and coordinate spaces.

\begin{algorithm}[t]
\small
\caption{Iterative update for antibody-antigen binding interface}
\label{alg:iterative_update}
\begin{algorithmic}[1]
\renewcommand{\algorithmicrequire}{\textbf{Input:}}
\renewcommand{\algorithmicensure}{\textbf{Output:}}
\REQUIRE{Initial embeddings, coordinates, number of iterations $T$, number of layers $t$}
\ENSURE{Updated embeddings and coordinates}
\FOR{$T = 1$ to $3$}
    \STATE{Construct intra-graph and inter-graph according to Appendix~\ref{appendix-sec:graph-construct} using embeddings and coordinates calculated}
    \FOR{$t = 1$ to $3$}
    \IF{$t = 1$}
    \STATE{Triangle multiplicative update for paratope}
    \ENDIF
    \STATE{Dual EMP update}
   \ENDFOR
    \STATE{Update the epitope coordinates and embeddings}
    \IF{$T = 3$}
        \STATE{Calculate training loss}
    \ENDIF
\ENDFOR
\end{algorithmic}
\end{algorithm}

\subsection{Loss Function}
\label{appendix-sec:loss}

Our training objective combines multiple loss terms to ensure accurate prediction of both sequence and structure. 

{\bf Sequence Loss.} The Sequence Negative Log Likelihood (SNLL) loss evaluates sequence prediction accuracy for masked positions in CDR regions:
\begin{equation}
\label{eq:loss-seq}
\begin{aligned}
    \mathcal{L}_{\textit{seq}} &= -\frac{1}{|V_p| \cdot | \mathcal{R}|}\sum_{v_i \in V_p } \sum_{r \in \mathcal{R}} y_{i,r} \log(\textit{softmax}(\textit{MLP}(\vect{Z}_i))_r),
\end{aligned}
\end{equation}
where $V_p$ is the set of masked residues in epitope-binding CDRs, $\vect{Z}_i$ is the output embedding of residue $v_i$, $\mathcal{R}$ is the set of all possible residue types,
$y_{i,r}$ indicates the ground truth residue type through one-hot encoding. $\textit{MLP}(\vect{Z}_i)$ outputs a vector of logits for each residue type and $\textit{softmax}(\textit{MLP}(\vect{Z}_i))_r$ gives the predicted probability for each residue type $r \in \mathcal{R}$.

{\bf Structural Loss.}
For structural accuracy, we employ a coordinate loss computed over the full antibody structure $\vect{X}_{ab}$:
\begin{equation*}
    \mathcal{L}_{\textit{coord}} = \frac{1}{|V_{ab}|} \sum_{v_i \in V_{ab}} \ell_{\textit{huber}}(\vect{X}_i^{\textit{pred}} - \vect{X}_i^{\textit{true}}T)
\end{equation*}
using the Huber loss function:
\begin{equation*}
    \ell_{\textit{huber}}(x) = 
    \begin{cases}
        0.5x^2, & \textit{if } |x| < \delta \\
        \delta(|x| - 0.5\delta), & \textit{otherwise}
    \end{cases}.
\end{equation*}
where $T$ represents the optimal rigid transformation obtained through Kabsch alignment algorithm \cite{kabsch1976rotation} using backbone atoms (N, C$\alpha$, C, O) of each residue, 
$\vect{X}_i^{\textit{pred}}$ and $\vect{X}_i^{\textit{true}}$ are the predicted and ground-truth 3D coordinates for residue $v_i$,
$\delta=1$ is the threshold parameter.

To maintain chemical validity, we include a backbone bond loss:
\begin{equation*}
    \mathcal{L}_{\textit{bond}} = \frac{1}{|\mathcal{A}|} \sum_{b \in \mathcal{A}} \ell_{\textit{huber}}(b^{\textit{pred}} - b^\textit{true})
\end{equation*}
where $\mathcal{A}$ contains all atoms in the antibody,
$b^{\textit{pred}}$ and $b^\textit{true}$ are the predicted and ground-truth bond length derived from $\vect{X}^{\textit{pred}}$
and $\vect{X}^{\textit{true}}$, respectively.
The final structure loss is the combination of the coordinate loss and the backbone bond loss:
\begin{equation}
\label{eq:loss-struct}
    \mathcal{L}_{\textit{struct}} = \mathcal{L}_{\textit{coord}} + \mathcal{L}_{\textit{bond}}
\end{equation}

{\bf Interface Loss.}
The interface loss evaluates the quality of predicted interactions between antibody paratope and antigen epitope regions using $\vect{X}_{inter}$. This loss comprises two components: a structural accuracy term and an edge distance prediction term. The structural accuracy loss measures the deviation of predicted atomic coordinates from their true positions in the paratope region:
\begin{equation*}
    \mathcal{L}_{\textit{struct}} = \frac{1}{|\mathcal{P}|} \sum_{i \in \mathcal{P}} \ell_{\textit{huber}}(\vect{X}_k^{\textit{pred}} - \vect{X}_k^{\textit{true}}),
\end{equation*}
where $\mathcal{P}$ denotes the set of atoms in the antibody paratope region, and $\vect{X}_{k}^{\textit{pred}}$ and $\vect{X}_{k}^{\textit{true}}$ denote the predicted and true 3D coordinates for atom $k$, respectively.
The edge distance prediction loss assesses the accuracy of predicted residue interactions within the binding interface:
\begin{equation*}
    \mathcal{L}_{\textit{edge}} = \frac{1}{|E_{\textit{inter}}|} \sum_{(v_i,v_j) \in E_{\textit{inter}}} \ell_{\textit{huber}}(f_{\theta}(\vect{h}_i, \vect{h}j) - dist_{ij}^{*}),
\end{equation*}
where $E_{\textit{inter}}$ is the set of epitope-paratope residue pairs at the interface with $v_i \in V_{ae}$ from the epitope and $v_j \in V_p$ from the paratope. The network $f_{\theta}$ predicts inter-residue distances based on residue embeddings $\vect{H}_i$ and $\vect{H}_j$ as input, $dist_{ij}^{*}$ is computed directly from coordinates as the minimum distance between atoms, $\vect{h}_i$ and $\vect{h}_j$, which are compared against true distances $d{ij}^*$ computed from atomic coordinates.

The total interface loss combines these components:
\begin{equation}
\label{eq:loss-interface}
    \mathcal{L}_{\textit{interface}} = \mathcal{L}_{\textit{paratope}} + \mathcal{L}_{\textit{edge}}
\end{equation}

{\bf Final Loss Function.}
By combining Equations \ref{eq:loss-seq}-\ref{eq:loss-interface}, the final loss function is calculating following:
\begin{equation}
  \mathcal{L}_{\textit{total}} = \mathcal{L}_{\textit{seq}} + \mathcal{L}_{\textit{struct}} + \mathcal{L}_{\textit{interface}}.
\end{equation}
This comprehensive loss function ensures accurate prediction of both sequence and structure while maintaining chemical validity and interface quality.

\section{Additional Experimental Settings}

\subsection{Datasets}
\label{appendix-sec:dataset}
Our antibody dataset preprocessing pipeline leverages the Structural Antibody Database (SAbDab) snapshot from November 2022, with structures pre-numbered using the IMGT numbering scheme \cite{lefranc2003imgt}. The IMGT system provides precise residue numbering for antibody structures, defining specific ranges for framework regions (FR) and complementarity determining regions (CDRs) in both heavy and light chains. For the heavy chain, these ranges encompass FR-H1 (1-26), CDR-H1 (27-38), FR-H2 (39-55), CDR-H2 (56-65), FR-H3 (66-104), CDR-H3 (105-117), and FR-H4 (118-129), with parallel definitions for the light chain regions.
This standardized numbering scheme is crucial for consistent processing and analysis across all antibody structures.

We implement stringent filtering criteria, focusing exclusively on protein and peptide antigens while ensuring no chain overlap between antibody and antigen components. Other potential antigen types like small molecules, nucleic acids, or carbohydrates are excluded from our dataset. The validation process enforces three critical requirements: structural completeness, with verification of conserved residues (CYS23, CYS104, TRP41) in both chains; correct IMGT numbering across all regions; and comprehensive atomic coordinate validation, including backbone and sidechain atoms, with a 6.6 \AA\ contact distance threshold for residue interactions.

The final dataset preprocessing stage involves sequence-based clustering using MMseqs2 with a 40\% identity threshold, particularly focusing on CDR-H3 sequence similarity. This process yields 3,246 training and 365 validation antibodies. The curated SAbDab dataset serves as the foundation for both the RAbD \cite{dunbar2013sabdab} (60 PDB structures) and the IgFold \cite{ruffolo2023IgFold} benchmark dataset (51 structures), ensuring consistent quality and standardization across all experimental evaluations.

\subsection{Implementation Details}
\label{appendix-sec:exp-setting}
We implement Igformer using PyTorch \cite{paszke2019pytorch} and train the model on a single GeForce RTX 3090 Ti GPU using the Adam optimizer \cite{kingma2015adam}. The model training process utilizes a batch size of 16, running for 150 epochs in task 1 and 300 epochs for tasks 2-4, while preserving the 10 best-performing model checkpoints throughout the training phase of each individual task. The model architecture employs a 64-dimensional embedding space for antibody feature representation, 14 channels (i.e., 14 atoms) for each residue, and 3 iterations. The graph structure connects each residue to its 9 nearest neighbors (k = 9). The model incorporates a dropout rate of 0.1 (dropout = 0.1) to enhance generalization. 
Training proceeds through 3 iterations per batch, with each iteration updating epitope coordinates and embeddings, followed by reconstruction of intra-graph and inter-graph connections. The final iteration is used for loss calculation.
For structure analysis, the model defines interacting residues in the antibody-antigen interface using a binding distance threshold of 6.6 \AA. 

The optimization process employs gradient clipping at 1.0 and implements an adaptive learning rate schedule. The learning rate decays exponentially from 0.0013 to 0.0001 across the training epochs, with the decay factor determined as $ln(0.0001/0.0013)$, divided by the total number of training steps, ensuring a smooth convergence to the target learning rate.

\textbf{Feature Initiation.}
Igformer employs two parallel embedding components to construct the initial residue representation. 
First, the residue type embedding transforms each amino acid into a 64-dimensional vector using a learned embedding matrix of size $25\times 64$, where 25 accounts for the 20 types of amino acids and special tokens. Each amino acid is mapped to a unique embedding vector $\vect{H}_i^{\textit{res}}$ that captures its chemical and physical properties.
The positional information is encoded through a separate learned embedding matrix of size $192 \times 64$, which maps each position in the sequence to a 64-dimensional vector $\vect{H}_i^{\textit{pos}}$. This enables the model to maintain sequential context for sequences up to 192 residues in length, allowing differentiation between identical amino acids at different positions.
The final feature representation $\vect{H}_i$ is computed through element-wise addition of these two embeddings. This additive combination maintains the 64-dimensional feature space while integrating both residue identity and positional context, enabling the model to learn position-aware representations with comprehensive chemical information.

\textbf{EMP Module.}
The EMP module iterates three times (3 blocks) and consists of several key components: feature initialization, feature similarity network, EdgeMLP, CoordMLP and NodeMLP with residual connections, as detailed in Table \ref{tab:igformer_layers}. 
Each layer incorporates dropout \cite{srivastava2014dropout} mechanisms after activation functions. The similarity features are normalized after being mapped from 196 to 128 dimensions, including $sim_{ij}^{res}$ and $sim_{ij}^{atom}$. 
The iterative application of these components ensures efficient updates of both node features and coordinate transformations while maintaining stable training dynamics through the dropout mechanism.
The EdgeMLP processes concatenated input features of dimension 384, comprising source\_node features (128), target\_node features (128), and similarity features (128). This design enables the model to effectively capture interactions between residues while preserving both structural and relational information crucial for antibody representation.

\begin{table}[t]
    \centering
    \small
    \caption{Structure of the EMP module.}
    \label{tab:igformer_layers}
    \begin{tabular}{l|l}
        \hline \hline
        \textbf{Component} & \textbf{Layer Structure} \\
        \hline
        Initialization Network & Linear(64 $\rightarrow$ 128), Dropout(0.1) \\
        \hline
        Radial Feature Network & Linear(196 $\rightarrow$  128), Post-normalization \\
        \hline
        \multirow{3}{*}{EdgeMLP} & Input: source\_features[128] $\oplus$ target\_features[128] $\oplus$ radial\_features[128] \\
                         & Linear(384 $\rightarrow$ 128), SiLU(), Dropout(0.1) \\
                         & Linear(128 $\rightarrow$ 128), SiLU(), Dropout(0.1) \\
        \hline
        \multirow{2}{*}{CoordMLP}         & Linear(128 $\rightarrow$ 128), SiLU(), Dropout(0.1) \\
                         & Linear(128 $\rightarrow$ 3) \\
        \hline
        \multirow{2}{*}{NodeMLP}          & Linear(64 $\rightarrow$ 128), SiLU(), Dropout(0.1) \\
                         & Linear(128 $\rightarrow$ 128), SiLU(), Dropout(0.1) \\
        \hline \hline
    \end{tabular}
\vspace{-4mm}
\end{table}

\begin{table}[t]
    \centering
    \small
    \caption{Structure of the inter-grpah refinement module.}
    \label{tab:additional_layers}
    \begin{tabular}{l|l}
        \hline \hline
        \textbf{Component} & \textbf{Layer Structure} \\
        \hline
        APP & Propagation with residue connection (128 $\rightarrow$ 128) \\
        \hline
        \multirow{3}{*}{SGFormer Layer (5 blocks)} & Input: Linear(128 $\rightarrow$ 64) \\
                             & Query: Linear(64 $\rightarrow$ 64), Key: Linear(64 $\rightarrow$ 64), Value: Linear(64 $\rightarrow$ 64) \\
                             & LayerNorm, Residual ($\alpha = 0.9$), Dropout(0.5) \\
        \hline
        Output & Linear(64 $\rightarrow$ 128) \\
        \hline \hline
    \end{tabular}
\vspace{-4mm}
\end{table}

\begin{table}[t]
    \centering
    \small
    \caption{Structure of the triangle module.}
    \label{tab:feature_pair_layers}
    \begin{tabular}{l|l}
        \hline \hline
        \textbf{Component} & \textbf{Layer Structure} \\
        \hline
        \multirow{2}{*}{Feature Pair MLP} & Input: features\_i[128] $\oplus$ features\_j[128] \\
                         & Linear(256 $\rightarrow$ 256), SiLU(), Dropout(0.1) \\
        \hline
        \multirow{2}{*}{Triangle Multiply (outgoing/incoming)} & Input: Matrix($n \times n \times 256$) \\
                                            & Linear(256 $\rightarrow$ 256), SiLU(), Linear(256 $\rightarrow$ 256) \\
        \hline
        \multirow{3}{*}{Triangle Attention (outgoing/incoming)} & Input: Matrix($n \times n \times 256$) \\
                                             & Query: Linear(256 $\rightarrow$ 256), Key: Linear(256 $\rightarrow$ 256), Value: Linear(256 $\rightarrow$ 256) \\
                                             & LayerNorm() \\
        \hline
        \multirow{2}{*}{Dimension Reduction} & Input: Diagonal($n \times 256$) \\
                            & Linear(256 $\rightarrow$ 128), LayerNorm() \\
        \hline \hline
    \end{tabular}
\vspace{-4mm}
\end{table}

\textbf{APP and SGFormer.}
As shown in Table \ref{tab:additional_layers}, the inter-grpah refinment network consists of two main components. First, the APP module conducts 16 steps of propagation with $\alpha=0.1$, , allocating 90\% weight to neighbor-propagated information while retaining 10\% of original residue features. 
This enables effective information spread across the graph while preserving essential original residue characteristics.
The SGFormer component employs 5 single-head attention layers, computing attention once per layer rather than utilizing multiple attention heads. In its residual connections after each transformer layer, it uses $\alpha=0.9$, emphasizing transformed features (90\%) while maintaining minimal previous layer information (10\%).
Layer normalization is applied to the 64-dimensional feature space after each attention computation, working in conjunction with dropout (0.5) to regulate feature magnitudes and ensure training stability.

\textbf{Triangle Module.}
The triangle network consists of two main components, as detailed in Table \ref{tab:feature_pair_layers}. First, the triangle multiplication module performs feature pair-wise propagation with directed connections.
For each sequence of $n$ residues, it constructs an $n\times n$ matrix by concatenating paired residue features (128-dim) and projecting them to 256-dimensional space via an MLP, establishing directed connections where $\vect{Z}_{ij}$ $\neq$ $\vect{Z}_{ji}$.
The triangle attention component uses four sequential operations: multiply outgoing, multiply incoming, attention outgoing, and attention incoming. Each operation maintains the $n \times n \times 256$ dimensionality while enabling distinct information flows. The multiply operations use MLPs for feature transformation while preserving the triangle structure, whereas the attention operations employ single-head self-attention with query/key/value projections to capture long-range dependencies. The final step extracts diagonal elements and reduces them to 128 dimensions for network compatibility.
The implementation incorporates a dropout rate of 0.1 and SiLU activation functions across all linear projections. Layer normalization is applied after attention operations, and an additional normalization step in the final dimension reduction. The entire module operates on a per-sequence basis, processing each unique sequence independently before the features are reassembled.

\textbf{Dual EMP Module.}
The Dual EMP architecture extends the base EMP module by operating on two distinct graph structures simultaneously. It processes the intra-graph and inter-graph representations, along with their respective coordinates ($\vect{X}_{full}$ and $\vect{X}_{inter}$), while maintaining shared embeddings between them. A detailed description of this architecture is provided in Section~\ref{sec:subsec-update}.

\textbf{Embedding Distance.}
As detailed in Table \ref{tab:edge_distance_prediction_mlp}, The edge distance MLP computes pairwise distances between residue embeddings. First, it concatenates
pairs of 128-dimensional residue embeddings into 256-dimensional vectors. These concatenated features are processed through a two-layer MLP with SiLU activations. The final linear projection produces a scalar value representing the predicted pairwise distance or interaction strength.

\begin{table}[t]
    \centering
    \small
    \caption{Structure of the edge distance MLP and prediction MLP.}
    \label{tab:edge_distance_prediction_mlp}
    \begin{tabular}{l|l}
        \hline  \hline
        \textbf{Component} & \textbf{Layer Structure} \\
        \hline
        \multirow{2}{*}{Edge Distance MLP} & Input: residue\_i[128] $\oplus$ residue\_j[128]) \\
                          & SiLU(), Linear(256 $\rightarrow$ 128), SiLU(), Linear(128 $\rightarrow$ 1) \\
        \hline
        \multirow{2}{*}{Prediction MLP} & Input: hidden features[128] \\
                      & SiLU(), Linear(128 $\rightarrow$ 128), SiLU(), Linear(128 $\rightarrow$ 20) \\
        \hline \hline
    \end{tabular}
\vspace{-4mm}
\end{table}

\textbf{Residue Type Prediction.}
As shown in Table \ref{tab:edge_distance_prediction_mlp}, the residue type prediction MLP converts the learned 128-dimensional residue representations into probability distributions over 20 amino acid classes through a two-layer MLP with skip connections. The architecture processes features sequentially: an initial SiLU activation and linear projection preserve the 128-dimensional space, followed by a second SiLU activation and final linear transformation that maps to 20-dimensional amino acid class probabilities.

\subsection{Baselines}
\label{appendix-sec:baselines}
For the baselines, we adopt the hyperparameters and training procedures from their official releases since all methods utilize SAbDab to construct training sets of similar scale and distribution. We save model parameters from the top 10 validation rounds, compute metrics for each model independently, and report the mean values as final results. State-of-the-art models on each task are included as competitors. 

In Task 1, we evaluate against RosettaAb, an energy-based method using statistical functions \cite{adolf2018rabd}; DiffAb, a diffusion-based generative model \cite{luo2022diffab}; MEAN, an equivariant attention network \cite{kong2023mean}; HERN, an end-to-end framework limited to CDR-H3 \cite{jin2022hern}; and dyMEAN, an adaptive multi-channel equivariant network \cite{kong2023dymean}.

For Task 4, we compare Igformer against IgFold$\Rightarrow$HDock, IgFold$\Rightarrow$HERN, GT$\Rightarrow$HERN, and dyMEAN, each employing different docking strategies. IgFold$\Rightarrow$HDock follows a two-stage pipeline where IgFold first predicts the antibody structure \cite{ruffolo2023IgFold}, and HDock \cite{yan2020hdock}, using knowledge-based scoring functions, performs docking. IgFold is a specialized version of AlphaFold tailored for antibody structure prediction, while HDock applies traditional docking approaches. In contrast, IgFold$\Rightarrow$HERN adopts a similar two-stage approach but replaces HDock with HERN \cite{jin2022hern}, which takes the IgFold-predicted backbone structure as input and generates docked backbones, followed by Rosetta for side-chain packing \cite{adolf2018rabd}. GT$\Rightarrow$HERN further extends this pipeline by using ground truth antibody structures instead of predicted ones, allowing HERN to dock CDR-H3 along with other regions toward the epitope, thus providing an upper bound on performance by leveraging perfect structural information. Unlike these multi-stage pipelines, dyMEAN distinguishes itself as a fully end-to-end approach that directly models antigen-antibody interactions without requiring separate structure prediction and docking stages \cite{kong2023dymean}.

We adopt hyperparameters and training procedures from their official releases, as all methods utilize SAbDab for training. Specifically, HERN by Jin et al. \cite{jin2022hern} uses a larger hidden size of 256 with four layers and 16 RBF kernels for distance embedding. DiffAb by Luo et al. \cite{luo2022diffab} employs a hidden size of 128, a pair feature size of 64, six layers, and 100 diffusion steps. MEAN and dyMEAN share similar parameters, including an embedding size of 64, a hidden size of 128, three layers, three iterations, and nine nearest neighbors. dyMEAN introduces an additional parameter, d = 16, for atom type and position embedding. All baseline models use nine neighbors for KNN graph construction.

\subsection{Metrics}
\label{appendix-sec:metrics}
We employ a comprehensive set of metrics to evaluate model performance in both antibody 1D sequence and 3D structure prediction tasks, encompassing both structural accuracy and sequence recovery:
\begin{itemize}[topsep=0.5mm, partopsep=0pt, itemsep=0pt, leftmargin=10pt]
    \item DockQ \cite{basu2016dockq}: This metric evaluates antibody-antigen interaction quality, especially within the paratope-epitope binding region. While the model optimizes paratope-epitope interactions within a specific cutoff distance during training, its evaluation considers interactions with the complete antigen structure, ensuring a comprehensive assessment of binding interface prediction.
    \item Root Mean Square Deviation (RMSD): This metric measures structural accuracy by calculating unaligned distances between predicted and actual CA atoms in CDR regions, providing direct assessment of local structural precision.
    \item TM-score \cite{zhang2005tmalign}: This metric assesses global structural similarity between predicted and reference antibody structures, evaluating the overall quality of structure prediction across the entire antibody.
    \item Local Distance Difference Test (lDDT) \cite{mariani2013lddt}: This metric provides an atomic-level assessment of local structural accuracy by comparing predicted and actual atomic positions across the entire antibody, offering detailed insights into structural fidelity.
    \item Amino Acid Recovery (AAR): This metric calculates the fraction of correctly predicted amino acids across the entire sequence, measuring overall sequence prediction accuracy.
    \item Contact AAR \cite{ramaraj2012caar}: This metric specifically evaluates prediction accuracy for residues in direct contact with the epitope, defined by a 6.6 \AA\ distance threshold, focusing on binding interface accuracy.
\end{itemize}

\section{Additional Experimental Results}
\label{appendix-sec:exp-results}

\begin{table}[t]
\centering
\small
\caption{\textit{Comparison of different paratopes.}}
\label{tab:additional-paratope}
\begin{tabular}{lcccc}
\hline \hline
\multicolumn{5}{c}{\bf Task 1} \\
\hline
Paratope & TMscore$\uparrow$ & lDDT$\uparrow$ & RMSD$\downarrow$ & DockQ$\uparrow$ \\
\hline
H3 & {\bf 0.9757} & {\bf0.8650} & {\bf 7.15} & {\bf 0.45} \\
H3 + L3 & 0.9734 & 0.8512 & 8.05 & 0.41 \\
\hline
\multicolumn{5}{c}{\bf Task 2} \\
\hline
H3 &   0.9706 & 0.8195 & - & 0.4255 \\
H3 + L3 & {\bf 0.9750} & {\bf 0.8311} & - & {\bf 0.4817} \\
\hline
\multicolumn{5}{c}{\bf Task 3} \\
\hline
H3 & {\bf 0.9681} & {\bf 0.7580} & - & {\bf 0.4600} \\
H3 + L3 & 0.9667 & 0.7490 & - & 0.4367 \\
\hline
\multicolumn{5}{c}{\bf Task 4} \\
\hline
H3 & {\bf 0.9730} & {\bf 0.8677} & {\bf 7.88} & {\bf 0.522} \\
H3 + L3 & 0.9710 & 0.8631 & 8.25 & 0.4931 \\
\bottomrule
\end{tabular}
\vspace{-4mm}
\end{table}

\subsection{Additional Ablation Study}
\label{appendix-sec:subsec-ablation}
{\bf Choice of Paratope.}
We investigate the impact of paratope selection by comparing models using CDR-H3 alone versus both CDR-H3 and CDR-L3 regions. As shown in Table \ref{tab:additional-paratope}, the choice of paratope influences model performance across different tasks.
For Task 1, using only CDR-H3 achieves superior performance with a TMscore of 0.9757, lDDT of 0.8650, RMSD of 7.15, and DockQ score of 0.45, compared to 0.9734, 0.8512, 8.05, and 0.41 respectively when including both CDR-H3 and L3. This trend continues in Task 3, where CDR-H3 alone yields better TMscore (0.9681 vs 0.9667), lDDT (0.7580 vs 0.7490), and DockQ (0.4600 vs 0.4367). Task 4 follows a similar pattern with improvements across all metrics when using only CDR-H3. 
Interestingly, Task 2 exhibits different characteristics, where including both CDR-H3 and CDR-L3 leads to improved performance, with TMscore increasing from 0.9706 to 0.9750, lDDT from 0.8195 to 0.8311, and DockQ from 0.4255 to 0.4817. These results underscore the importance of task-specific paratope selection for optimal performance.

{\bf Ablation Study on Task 4.}
We conduct additional experiments to evaluate the individual contributions of key components in Igformer for complex structure prediction. Table~\ref{tab:ablation-appendix} presents the performance impact of removing each architectural component. Our observations are as follows. Both the triangle multiplicative module and axial attention prove critical to model performance. The removal of TM leads to the most significant degradation, with RMSD increasing to 8.12 and DockQ dropping to 0.447, while removing AA results in an RMSD of 8.06 and DockQ of 0.453. Additionally, replacing the dual EMP architecture with a single message passing framework diminishes performance, with RMSD rising to 7.99 and DockQ decreasing to 0.452. These systematic evaluations further validate the essential contribution of each component in Igformer.

\begin{table}[t]
\vspace{-2mm}
\caption{Ablation study on complex structure prediction.}
\small
\centering
\scalebox{1.1}{
\setlength\tabcolsep{2pt}
\label{tab:ablation-appendix}
\begin{tabular}{lcc|cc}
\hline \hline
\multirow{2}{*}{Model} & \multicolumn{2}{c}{Generation} & \multicolumn{2}{c}{Docking} \\
\cline{2-5}
 & TMscore$\uparrow$ & lDDT$\uparrow$ & RMSD$\downarrow$ & DockQ$\uparrow$ \\
\hline
Igformer & \textbf{0.9730} & \textbf{0.8677} & \textbf{7.88} & \textbf{0.522} \\
- APP & 0.9718 & 0.8603 & 7.95 & 0.449 \\
- SGFormer & 0.9725 & 0.8614 & 7.92 & 0.453 \\
- TM & 0.9709 & 0.8543 & 8.12 & 0.447 \\
- AA & 0.9708 & 0.8565 & 8.06 & 0.453 \\
- Dual EMP & 0.9725 & 0.8611 & 7.99 & 0.452 \\
\hline \hline
\end{tabular}
}
\vspace{-4mm}
\end{table}

\subsection{Hyperparameter analysis}
\label{appendix-sec:subsec-hyperparameter}

As illustrated in Figure \ref{fig:parameters}, we analyze the impact of learnable weight factor $w$ (refer to Section~\ref{sec:subsec-emp}) on similarity computation across Tasks 1-4. 
The DockQ, as a primary performance metric, consistently achieves its peak at $w = 0.2$ across all tasks. This optimal value indicates that atom-level similarity plays a predominant role in model performance.
While both residue-level similarity (parameterized by $w$) and atom-level similarity contribute to modeling structural interactions, the optimal performance at $w = 0.2$ suggests that emphasizing atomic information (0.8) over residue-level features (0.2) most effectively captures antibody-antigen interactions. This finding provides crucial guidance for optimizing similarity matrix parameters in antibody sequence design and structure prediction tasks.

\begin{figure}[t]
    \centering
    \includegraphics[height=55mm]{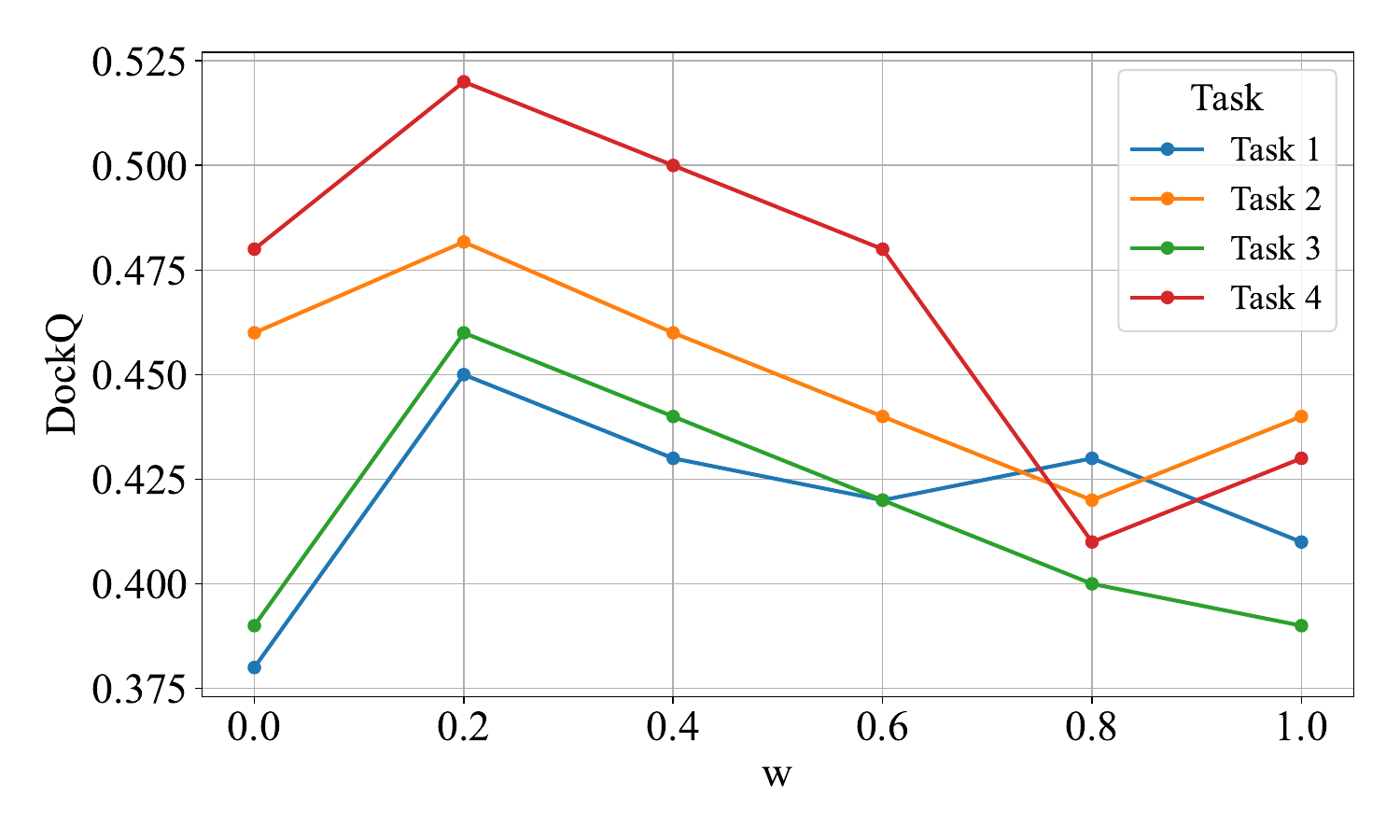}
    \vspace{-6mm}
    \caption{DockQ scores for varying $w$ values across four tasks, with the peak observed at $w=0.2$, emphasizing the importance of geometric distance in similairty matrix computation.}
    \vspace{-4mm}
    \label{fig:parameters}
\end{figure}

\subsection{Additional Case Study of Task 1}
\label{appendix-sec:subsec-case-study}
Figure~\ref{fig:fig6-appendix-case} illustrates more antibody structures generated by Igformer in Task 1.

\begin{figure}[t]
\centering
  \begin{small}
    \begin{tabular}{ccc}
        \includegraphics[height=50mm]{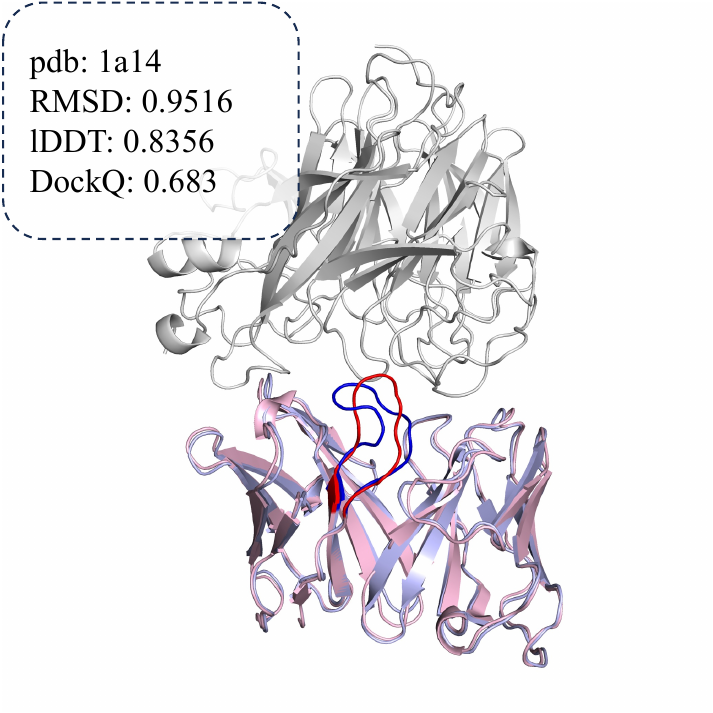} &
        \hspace{-3mm}
        \includegraphics[height=50mm]{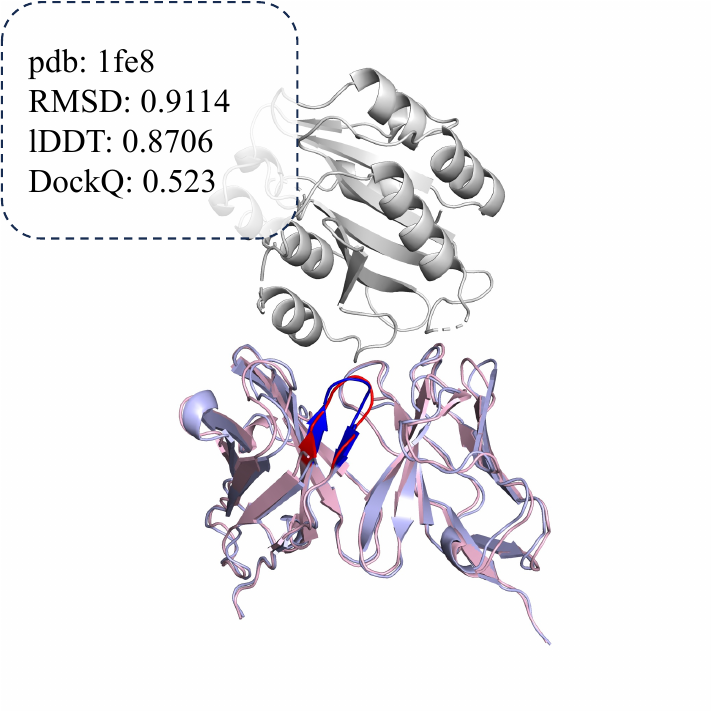} &
        \hspace{-3mm}
        \includegraphics[height=50mm]{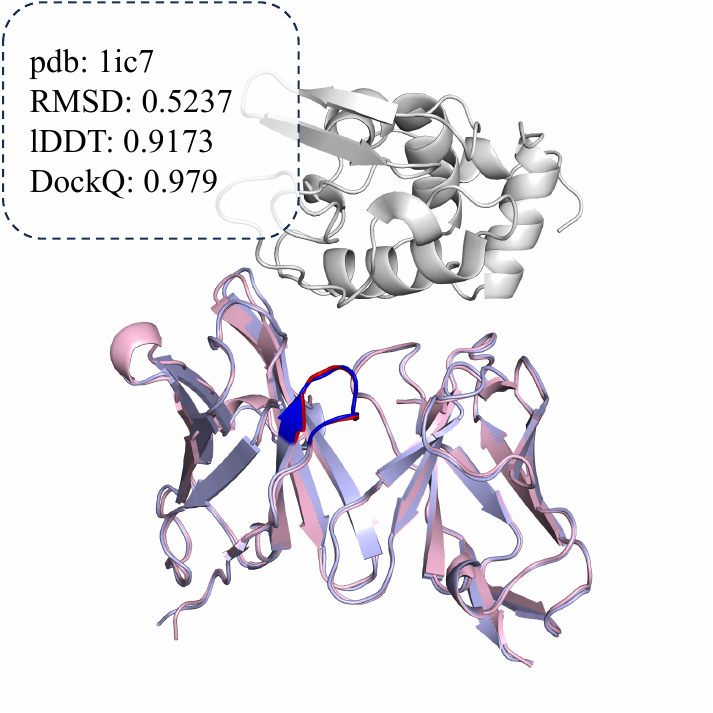} \\
        \includegraphics[height=50mm]{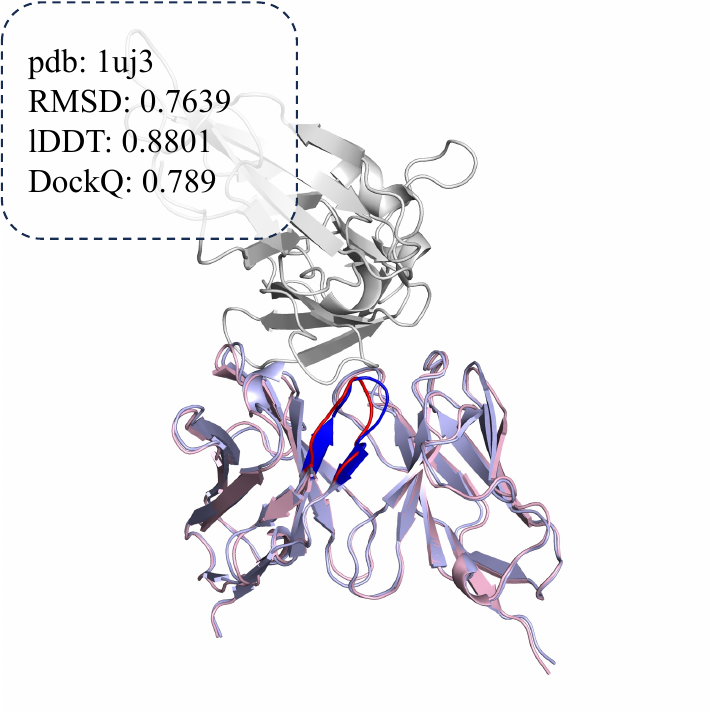} &
        \hspace{-3mm}
        \includegraphics[height=50mm]{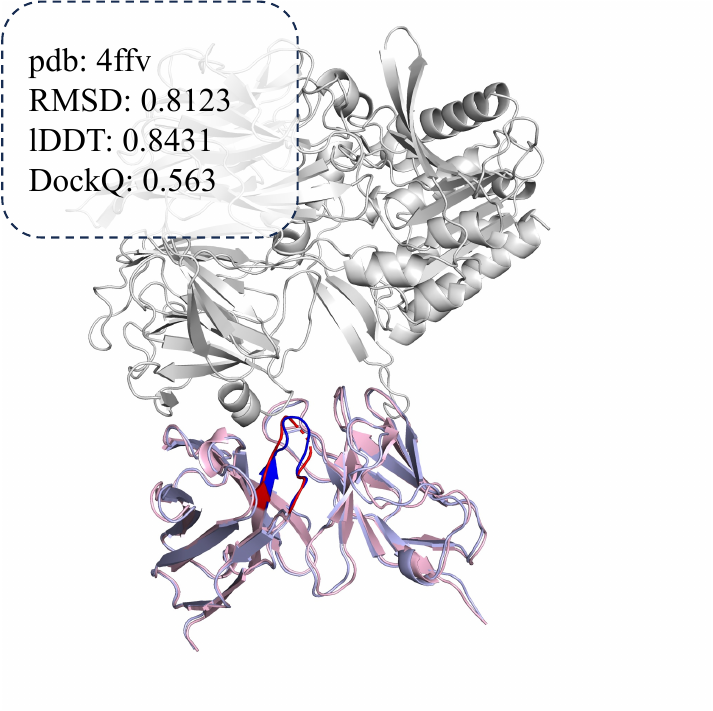} &
        \hspace{-3mm}
        \includegraphics[height=50mm]{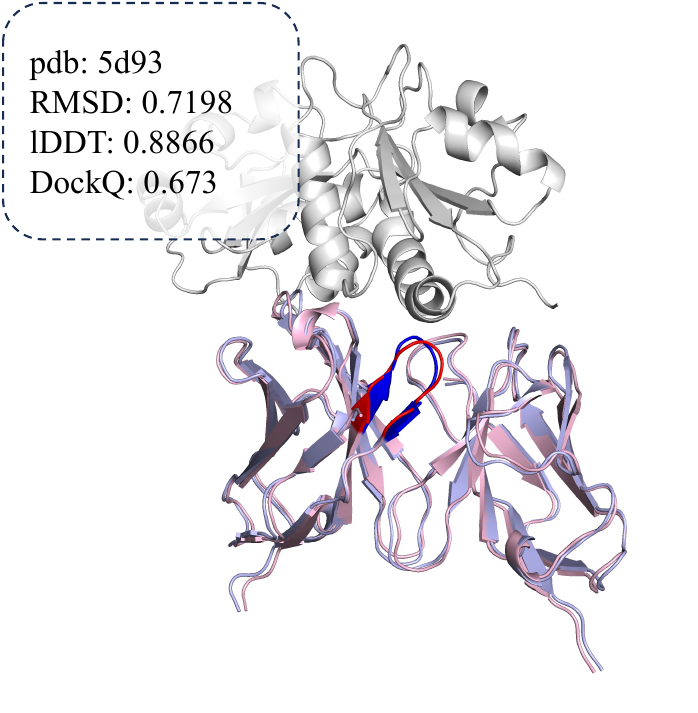} \\
    \end{tabular}
    \vspace{-2mm}
    \caption{Additional antibody structures generated by Igformer.}
    \label{fig:fig6-appendix-case}
    \vspace{-4mm}
  \end{small}
\end{figure}

\end{document}